\documentclass[a4paper,UKenglish,cleveref, autoref, thm-restate]{lipics-v2021}

\title{History-deterministic Vector Addition Systems}

\author{Sougata Bose}{University of Liverpool, UK}{sougata.bose@liverpool.ac.uk}{https://orcid.org/0000-0003-3662-3915}{}
\author{David Purser}{University of Liverpool, UK}{D.Purser@liverpool.ac.uk}{https://orcid.org/0000-0003-0394-1634}{}
\author{Patrick Totzke}{University of Liverpool, UK}{totzke@liverpool.ac.uk}{https://orcid.org/0000-0001-5274-8190}{}

\authorrunning{S.~Bose, D.~Purser and P.~Totzke} 

\Copyright{Sougata Bose, David Purser and Patrick Totzke}

\ccsdesc[500]{Theory of computation~Formal languages and automata theory}

\keywords{Vector Addition Systems, History-determinism, Good-for Games}

\category{} %

\relatedversion{This is the full version of a paper published in CONCUR 2023.}

\supplement{}%

\funding{
    We gratefully acknowledge support from the EPSRC, project number \texttt{EP/V025848/1}
    and the Royal Society, \texttt{IES\textbackslash R3\textbackslash 203109}.
}

\acknowledgements{The authors thank Ismaël Jecker, Piotr Hofman and Filip Mazowiecki for fruitful discussions.}%

\nolinenumbers %
\hideLIPIcs

\EventEditors{Guillermo A. P\'{e}rez and Jean-Fran\c{c}ois Raskin}
\EventNoEds{2}
\EventLongTitle{34th International Conference on Concurrency Theory (CONCUR 2023)}
\EventShortTitle{CONCUR 2023}
\EventAcronym{CONCUR}
\EventYear{2023}
\EventDate{September 18--23, 2023}
\EventLocation{Antwerp, Belgium}
\EventLogo{}
\SeriesVolume{279}
\ArticleNo{18}

\usepackage[basic]{complexity}
\usepackage{tabularx}

\usepackage{pgf}

\usepackage{stmaryrd}
\usepackage{adjustbox}
\usepackage{tikz}
\usetikzlibrary{arrows,automata}

\usepackage{xspace}

\usepackage{sansmath}

\usepackage{pgf}
\usepackage{tikz}

\newcommand{\+}[1]{\mathbb{#1}}
\newcommand{\?}[1]{\mathcal{#1}}

\usepackage{bm}
\renewcommand{\vec}[1]{\bm{#1}}
\newcommand{\support}[1]{\mathit{support}(#1)}

\newcommand{\x}{\times}
\newcommand{\N}{\+{N}}
\newcommand{\Z}{\mathbb{Z}}

\renewcommand{\epsilon}{\varepsilon}
\newcommand{\eps}{\varepsilon}

\newcommand{\hdness}{\text{HDness}\xspace}

\newcommand{\SYS}[4][]{\ensuremath{\if\relax\detokenize{#1}\relax\else#1\text{-}\fi\text{#2}\text{-VASS}^{#3}_{#4}}\xspace}
\newcommand{\Advass}[1][]{\SYS[#1]{D}{}{}}
\newcommand{\ARdvass}[1][]{\SYS[#1]{D}{0}{}}
\newcommand{\Avass}[1][]{\ensuremath{\if\relax\detokenize{#1}\relax\else#1\text{-}\fi\text{VASS}}\xspace}
\newcommand{\Ahvass}[1][]{\SYS[#1]{H}{}{}}
\newcommand{\Anvass}[1][]{\SYS[#1]{N}{}{}}
\newcommand{\ARhvass}[1][]{\SYS[#1]{H}{0}{}}

\newcommand{\AhvassE}[1][]{\SYS[#1]{H}{}{\eps}}

 \newcommand{\CLASS}[4][]{\ensuremath{\if\relax\detokenize{#1}\relax\else#1\text{-}\fi\mathcal{#2}^{#3}_{#4}}\xspace}
 \newcommand{\CDvass}[1][]{\CLASS[#1]{D}{}{}}
 \newcommand{\CHvass}[1][]{\CLASS[#1]{H}{}{}}
 \newcommand{\CNvass}[1][]{\CLASS[#1]{N}{}{}}
 \newcommand{\CHvassE}[1][]{\CLASS[#1]{H}{}{\eps}}
 \newcommand{\CNvassE}[1][]{\CLASS[#1]{N}{}{\eps}}
 \newcommand{\RDvass}[1][]{\CLASS[#1]{D}{0}{}}
 \newcommand{\RHvass}[1][]{\CLASS[#1]{H}{0}{}}
 \newcommand{\RNvass}[1][]{\CLASS[#1]{N}{0}{}}
 \newcommand{\RHvassE}[1][]{\CLASS[#1]{H}{0}{\eps}}
 \newcommand{\RNvassE}[1][]{\CLASS[#1]{N}{0}{\eps}}

\newcommand{\vass}{\Avass}

\newcommand{\dvass}{\Advass}
\newcommand{\hdvass}{\Ahvass}

\newcommand{\hdvasse}{\AhvassE}

\newcommand{\LnotDVASS}{\ensuremath{L_1}}
\newcommand{\LnotHDVASS}{\ensuremath{L_3}}
\newcommand{\LnotFSVASS}{\ensuremath{L_8}}

\newcommand{\Lanbn}{\ensuremath{L_5}}
\newcommand{\Lanblen}{\ensuremath{L_4}}
\newcommand{\Lanblenbarrier}{\ensuremath{L_7}}

\newcommand{\Lanbgen}{\ensuremath{L_2}}
\newcommand{\LmustVASSe}{\ensuremath{L_6}}

\newcommand{\Lnotrunion}{\ensuremath{L_{9}}}
\newcommand{\Lanbnastar}{\ensuremath{L_{13}}}
\newcommand{\Lnotcunion}{\ensuremath{L_{10}}}
\newcommand{\Lnotcint}{\ensuremath{L_{11}}}
\newcommand{\Lnotrint}{\ensuremath{L_{12}}}
\newcommand{\Lcross}{\ensuremath{L_{14}}}

\newcommand{\sys}[1]{\?{#1}}
\newcommand{\states}[1][]{Q_{#1}}
\newcommand{\fstates}[1][]{F_{#1}}
\newcommand{\alphabet}[1][]{\Sigma_{#1}}
\newcommand{\transitions}[1][]{\delta_{#1}}

\newcommand{\initialstate}{s_{0}}
\newcommand{\tlabel}{\mathit{label}}
\newcommand{\tfunction}[1]{\tlabel(#1)}
\newcommand{\effect}[1]{\mathit{effect}(#1)}

\newcommand{\vasstuple}{(\alphabet,\states,\transitions,\initialstate,\fstates)}
\newcommand{\step}[2][]{\xrightarrow[#1]{#2}}
\newcommand{\Post}[2][]{\mathit{Post}_{#1}({#2})}

\newcommand{\abs}[1]{\lvert#1\rvert}

\newcommand{\norm}[2][]{\lVert#2\rVert_{#1}}

\newcommand{\eqby}[2][=]{\stackrel{\mathit{{\scriptscriptstyle{#2}}}}{#1}}
\newcommand{\eqdef}{\eqby{def}}

\newcommand{\Lang}[2][]{\?L_{#1}(#2)}

\newcommand{\dec}{\text{dec}}
\newcommand{\inc}{\text{inc}}
\newcommand{\decop}[1]{X_{#1}{-}{-}}  % used in VASS when simulating counters
\newcommand{\incop}[1]{X_{#1}{+}{+}}
\newcommand{\ztest}{\text{ztest}}
\newcommand{\op}{\gamma}
\newcommand{\CMOPS}{\Gamma}

\begin{document}

\maketitle

\begin{abstract}
We consider history-determinism, a restricted form of non-determinism, for Vector Addition Systems with States (VASS) when used as acceptors to recognise languages of finite words.
History-determinism requires that the non-deterministic choices can be resolved on-the-fly; based on the past and
without jeopardising acceptance of any possible continuation of the input word.

Our results show that the history-deterministic (HD) VASS sit strictly between deterministic and non-deterministic VASS regardless of the number of counters. We compare the relative expressiveness of HD systems, and closure-properties of the induced language classes, with coverability and reachability semantics, and with and without $\epsilon$-labelled transitions.

Whereas in dimension 1, inclusion and regularity remain decidable, from dimension two onwards, HD-VASS with suitable resolver strategies, are essentially able to simulate 2-counter Minsky machines, leading to several undecidability results: It is undecidable whether a VASS is history-deterministic, or if a language equivalent history-deterministic VASS exists. Checking language inclusion between history-deterministic $2$-VASS is also undecidable.

 \end{abstract}

\section{Introduction}

Vector addition systems with states (VASSs)
are an established model of concurrency with extensive applications in modelling and analysis of hardware, software, chemical, biological and business processes.
They are non-deterministic finite automata equipped with a fixed number of integer counters
that may be incremented or decremented when changing control state, as long as they remain non-negative.

We explore the notion of \emph{history-determinism} %
for VASSs %
when used as acceptors to define languages of finite words.
History-determinism is a restricted form of non-determinism.
In a nutshell, a non-deterministic automaton is history-deterministic (HD)
if there exists a \emph{resolver}, which is a strategy to %
stepwise 
produce a run for any input word given one letter at a time,
in such a way that
if there exist some accepting run on the given word then the run produced by the resolver is also accepting.

The original motivation for HDness comes from formal verification:
most modelling formalisms incorporate some form of non-determinism,
e.g.,
to over-approximate deterministic algorithms,
to state specifications concisely,
or to model system behaviour due to uncontrollable external environments.
However, for non-deterministic models, many formal analysis techniques require costly determinisation steps that are often the main barrier to efficient procedures.
History-deterministic automata provide a middle ground:
they are typically more succinct, or even more expressive, than their deterministic counterparts
while preserving some of their good algorithmic properties.
They were also called
``good-for-games'' as they preserve the winner of games under composition and thus allow
solving games without determinisation.

Any resolver must always choose language-maximal successors, that is, the language of the chosen successor must also include the language of any alternative successor.
When considering languages of finite words, 
being able to continue making language-maximal choices is even 
a sufficient 
condition for being a resolver.
Therefore, in this case,
resolvers can be assumed to be (configuration) positional: they base their decisions only on the current configuration and not the full history leading to it.
Perhaps surprisingly, resolvers for VASSs are not necessarily monotone with respect to counter values,
and may require more than just comparing counters to integer thresholds (see \cref{app:lang-max}).

\vspace*{-0.3cm}
\subparagraph*{Related Work.}
VASSs, also known as Petri nets or partially blind counter automata,
have been studied intensively since their inception in the 1960s.
Early works focussed on modelling capabilities, relative expressiveness and closure properties of their recognised languages \cite{H1976,G1978,VVN1981,J1986}
but the bulk of research on VASSs concerns decidability and complexity of decision problems \cite{KM1969,L1976,R1978,K1982,JEM1999,JM1995,BGJ2014,L2021,CO2021}.
In order to define languages with VASSs, different definitions distinguish between coverability and reachability acceptance conditions, and whether or not silent ($\epsilon$) transitions %
are permitted.
Checking language emptiness amounts to testing coverability or reachability, which are \EXPSPACE~\cite{R1978,L1976} and Ackermann-complete~\cite{CO2021} respectively. Many other decision problems are undecidable, such as checking language inclusion, %
bisimulation and related equivalences
\cite{J2001} as well as checking (language) regularity \cite{JM1995}. 
Universality is undecidable for reachability acceptance \cite{VVN1981}
and decidable for coverability acceptance,
via a well-quasi-order argument but with extremely high
complexity (Hyper-Ackermannian in general \cite{JEM1999} and still Ackermannian in dimension $1$ \cite{HT2017}).
These negative results by and large rely on the presence of non-deterministic choice,
which motivates restricted forms of non-determinism
such as bounded ambiguity (that allows for decidable inclusion \cite{CH2022})
or the notion of history-determinism studied here.

VASS recognisable languages over infinite words are significantly more complex than their finite-word cousins,
both topologically and in terms of decision problems: %
already $1$-VASS with (cover) B\"uchi acceptance can recognise $\Sigma^1_1$-complete languages \cite{S2017,FS2014}
and have an undecidable universality problem \cite{BGHH2017}. Again, the added complexity is due to non-determinism (languages of deterministic models are Borel, lower in the analytical hierarchy).

\medskip
History-determinism was introduced independently, with slightly
different definitions, by Henzinger and Piterman~\cite{HP06} for solving games without determinisation, by Colcombet~\cite{Col09} for cost-functions, and by Kupferman, Safra, and Vardi~\cite{KSV06} for recognising derived tree languages of word automata.
These different definitions all coincide for finite automata \cite{BL19}
but not necessarily for more general quantitative automata \cite{BL21}.

Until now, history-determinism has mainly been studied for finite-state systems.
In this paper we continue a recent line of work 
\cite{GJLZ21,LZ22,EGJLZ22,HLT2022,BHLST2022,PT2022}
that studies the notion for infinite-state models
capable of recognising languages beyond ($\omega$-)regular ones.
For infinite-state systems, deterministic models are in general less expressive, not just less succinct, than their non-deterministic counterparts. 
In some cases they can be determinised, such is the case for quantitative automata~\cite{BL21} and timed automata
with safety and reachability acceptance \cite{HLT2022}.
In contrast, for pushdown automata~\cite{GJLZ21} and Parikh automata (VASS with $\Z$-valued counters; \cite{EGJLZ22}), 
and timed automata with co-B\"uchi acceptance, allowing history-determinism strictly increases expressiveness (and adds more closure properties) compared to the deterministic variant.
Whenever HD automata are strictly less expressive than fully non-deterministic ones, %
one can reasonably ask if 
there exists an equivalent HD automaton for a given
non-deterministic one.
This language HDness question is undecidable for pushdown and Parikh-automata
\cite{GJLZ21,EGJLZ22}.
In fact, even checking if a given (pushdown or Parikh) automaton is itself HD is undecidable
(for Parikh automata this follows for example by the undecidability of $2$-dim.~robot games \cite{NPR2016}). On the other hand, checking HDness 
for timed automata is decidable \cite{HLT2022} and various models of quantitative automata~\cite{BL22}.

Most closely related to our work is that of
Prakash and Thejaswini \cite{PT2022} who study history-deterministic one counter automata (OCA; PDA with unary stack alphabet) and nets (OCN; $1$-dimensional VASSs) with state-based (coverability) acceptance.
They show that checking automata HDness and inclusion are undecidable for OCA but remain decidable for OCNs.
A useful consequence of their construction is that for any OCN one can construct a language equivalent deterministic OCA (with zero-test),
albeit with a doubly exponential blow-up.
They do not consider closure properties and leave open whether history-deterministic OCNs 
can be determinised, are equally expressive as fully non-deterministic OCNs, or fall strictly in between in expressiveness.  %
Our work extends and generalises this paper in several directions.

\vspace*{-0.3cm}
\subparagraph*{Our Contributions.}
We study history-deterministic VASSs on finite words and without restricting the dimension.
We consider coverability and reachability acceptance conditions, with and without silent ($\eps$) transitions,
and in all cases study the relative expressiveness, closure properties, and related decision problems.

We show that HD VASSs are more expressive than deterministic, but less expressive than non-deterministic ones.
The same is true for languages recognised by VASSs of any fixed dimensions $k$,
which answers the open question in \cite{PT2022} for $k=1$.
In particular, we provide examples of $1$-dim.~HD VASSs for which no equivalent deterministic ones exist in any dimension $k$,
and also demonstrate that HD VASSs are strictly more expressive than finitely sequential ones (another restricted form of non-determinism).

We show that HD VASS languages are closed under inverse homomorphisms and intersections for both coverability and reachability semantics,
although sometimes necessarily increasing the dimension.
Coverability languages are closed under unions, whereas reachability languages are not. Neither are closed under other standard operations,
including complementation, concatenation, homomorphisms, iteration and commutative closures.

We report that HDness is not sufficient for decidability of inclusion checking, even for $2$-dimensional VASSs.
A direct consequence is the undecidability of checking HDness of a given $2$-VASS, contrasting decidability in dimension $1$.
Further, it is undecidable to check if a given VASS has a HD equivalent,
and also if a given HD VASS recognises a regular language.

\section{Definitions}
\label{sec:definitions}

\vspace*{-0.3cm}
\subparagraph*{Vector-Addition Systems and their recognised languages.}

    A $k$-dimensional \emph{vector-addition system} ($k$-VASS)
    is a non-deterministic finite automaton whose transitions
    manipulate $k$ non-negative integer counters.
    It is given by $\sys{A}=\vasstuple$ consisting of 
         a finite alphabet $\alphabet$; 
         a finite set of control states $\states$; 
         a transition relation $\transitions\subseteq \states\x\alphabet\cup\{\eps\}\x\Z^k\x\states$; 
         an initial state $\initialstate$; 
         a subset $\fstates\subseteq\states$ of final states.

    For a transition $t=(s,a,e,s')\in\transitions$ we sometimes write 
    $\tfunction{t}\eqdef a$ for the letter from $\alphabet\cup\{\eps\}$ it reads
    and
    $\effect{t}\eqdef e$ for its \emph{effect} on the counters.
    $\norm{\delta}$ denotes the largest absolute effect among all transitions on any counter.

    A VASS naturally induces an infinite-state labelled transition system in which each
    \emph{configuration} is a pair $(s,v)\in \states\x\N^k$ comprising a control state and a \emph{non-negative} integer vector.
    Every transition $t=(s,a,e,s')\in\transitions$ gives rise to
    steps $(s,v)\step{t}(s',v')$ for all $v,v'\in\N$ with $v'=v+e$.
    We will call a path $\rho=(s_0,v_0)\step{t_1}(s_1,v_1)\step{t_1}\ldots\step{t_k}(s_k,v_k)$
    a \emph{run} of the VASS and say it is a \emph{cycle} if $s_0=s_{k}$.
    Its \emph{effect} is the sum of all transition effects $\effect{\rho}\eqdef \sum_{i=1}^{k} \effect{t_i}$.
    A run $\rho$ as above %
    \emph{reads} the word $\tfunction{\rho}=\tfunction{t_1}\tfunction{t_2}\ldots\tfunction{t_k}\in\Sigma^*$.
    It is \emph{accepting} if it ends in a final configuration.

    We consider two different definitions for what constitutes a final (also \emph{accepting}) configurations:
    In the \emph{coverability} semantics, the set of final configurations is $F\x\N$.
    In the \emph{reachability} semantics, only configurations from $F\x\vec{0}$ are final.
We define the \emph{language}
$\Lang[\?A]{s,v}\subseteq\alphabet^*$
of a configuration $(s,v)$
to contain exactly all words read by some accepting run starting in $(s,v)$ (we omit the subscript $\?A$ if the VASS is clear from context).
For notational convenience, we will lift this to sets $S\subseteq\states\x\N^k$ of configurations in the natural way:
$\Lang[\?A]{S}\eqdef\bigcup_{(s,v)\in S}\Lang[\?A]{s,v}$ and define
the language of $\sys{A}$ as that of its initial state with all counters zero: $\Lang{\?A}\eqdef\Lang[\?A]{\initialstate,\vec{0}}$.

We will sometimes denote languages using short-hand ``counting expressions''. For instance, we write
$a^nb^{\le n}$ for the language
$\{a^nb^m\mid n\ge m\}$ over $\Sigma=\{a,b\}$.

\vspace*{-0.3cm}
\subparagraph*{Deterministic and finitely-sequential VASSs.}
A VASS 
$\?A=\vasstuple$ is
called $\eps$-free if 
no transition is labelled by $\eps$.
It is \emph{deterministic} if it is $\eps$-free and for every pair $(s,a) \in \states\x\alphabet$ there
is at most one transition $t=(s,a,e,s')\in\transitions$.
A VASS is \emph{finitely sequential} if it is the finite union of deterministic VASSs.
That is, all transitions from its initial state $\initialstate$ are labelled by $\eps$
and lead to an initial state of one of finitely many deterministic VASSs.

\vspace*{-0.3cm}
\subparagraph*{History-deterministic VASSs.}
A VASS is \emph{history-deterministic} if one can resolve non-deterministic choices on-the-fly.
More formally, consider a function
$r: (\states\x\N^k\x\transitions)^*(\states\x\N)\x\Sigma\to\transitions$ that, given a finite run 
$\rho_i = (s_0,v_0)\step{t_1}(s_1,v_1)\step{t_2}\dots
\step{t_i}(s_i,v_i)
$
and a next letter $a_{i+1}\in \Sigma$, returns a transition 
$r(\rho_i,a_{i+1})=t_{i+1} = (s_i,a_{i+1},e_{i+1},s_{i+1})\in\transitions$
with $v_i+e_{i+1}\in\N^k$.
This yields, for every word $w=a_1a_2\ldots \in \Sigma^*$
and initial configuration $(s_0,v_0)$, a unique run
in which the $i$th step
$(s_{i-1},v_{i-1})\step{t_i}(s_{i},v_{i})$
results from a transition chosen by $r$.
Such a function is called \emph{resolver} if for any input word
$w\in \Lang[\?A]{s_0,v_0}$
the constructed run $\rho$ from initial configuration
$(s_0,v_0)$ is accepting.
A $k$-VASS is \emph{history-deterministic} if such a resolver exists.

\vspace*{-0.3cm}
\subparagraph*{Language Classes.}
We denote by 
$\CDvass[k]$,
$\CHvass[k]$,and 
$\CNvass[k]$
the classes of languages recognised by $k$-dimensional $\eps$-free deterministic, history-deterministic, and fully non-deterministic VASSs, in the coverability semantics.
Similarly, let
$\RDvass[k]$,
$\RHvass[k]$, and 
$\RNvass[k]$
denote the classes of languages recognised by $k$-dimensional $\eps$-free deterministic, history-deterministic, and fully non-deterministic VASSs, in the reachability semantics.
Finally, define
$\CHvassE[k]$,$\CNvassE[k]$ $\RHvassE[k]$, and
$\RNvassE[k]$, as above but without the restriction to $\eps$-free systems.
When dropping the parameter $k$ we refer to the union over all dimensions $k$. For instance,
$\CHvass \eqdef\bigcup_{k\in\N}\CHvass[k]$.

\section{Expressiveness}
\label{sec:expressiveness}

We consider the hierarchy of language classes recognised by vector addition systems, varying definitions in three directions: the degree of non-determinism, reachability vs coverability acceptance, and with/without $\epsilon$-transitions. 

The situation is depicted in~\cref{fig:comparison}. 
We start by looking at the classes defined by $\eps$-free systems (in \cref{sub:exprnoneps})
before discussing the effect of $\epsilon$-transitions (in \cref{sub:expr:eps}) and following this up with a comparison with finitely-sequential VASS (in \cref{sub:expr:fs}).

\subsection{Separating determinism, history-determinism and non-determinism}
\label{sub:exprnoneps}

\begin{figure}
\begin{adjustbox}{max width=\textwidth}
\begin{tikzpicture}[->,>=stealth',auto,
                    semithick]
  \tikzstyle{vertex}=[circle,fill=black!25,minimum size=17pt,inner sep=0pt]

  \node (LD) at (-1.5,-2.2){\CDvass[k]};

  \node (LHD) at (-1.5,-1){\CHvass[k]};

  \node (L) at ({-1.5+2.5},{1}){\CNvass[k]};

  \node (eLHD) at ({-1.5-2.2},1){\CHvassE[k]};
  \node (eL) at ({-1.5},{4}){\CNvassE[k]};

  \node (LD0) at ({7},-2.2){\RDvass[k]};

  \node (LHD0) at ({7},-1){\RHvass[k]};

  \node (eLHD0) at ({7-2.5},2){\RHvassE[k]};
  \node (L0) at ({7+2.5},{2}){\RNvass[k]};
  \node (eL0) at ({7},{4}){\RNvassE[k]};
  ;

  \path 
  (LHD) edge[dotted] node[align=left]{$k=1$ eqiv\\ $k\ge 2:  \LmustVASSe$} (eLHD)
  (LHD) edge node[right,align=left]{$\LnotHDVASS$}  (L)
  (LHD0) edge node[] {$\Lanblenbarrier$} (eLHD0)
  (LHD0) edge node[below right] {$\Lanblen$}(L0)

  (LD0) edge node[right] {$\Lanbgen$}(LHD0)
  (LD) edge node[right] {$\LnotDVASS$}(LHD)

  (eLHD) edge  node[midway] {$\LnotHDVASS$} (eL)
  (eLHD0) edge node {$\LnotHDVASS$} (eL0)
  (L) edge[dotted] node[above right,align=left]{$k=1$ eqiv\\ $k\ge 2: \LmustVASSe$}(eL)
  (L0) edge node[auto,above right] {$\LmustVASSe$} (eL0)

  (-0,-1.6) edge[-,dashed,red] node[above,align=left] {$\ni \Lanblen \not\in$} (5.5,-1.6)

  (-0,-1.6) edge[-,dashed,red] (LD)
  (-0,-1.6) edge[-,dashed,red] (LHD)
  (5.5,-1.6) edge[-,dashed,red] (LHD0)
  (5.5,-1.6) edge[-,dashed,red] (LD0)

  (eLHD) edge[-,dashed,red] node[below,align=left] {$\ni \LmustVASSe \not\in$ \quad
  $\not\ni \LnotHDVASS \in$} (L)

  (eLHD0) edge[-,dashed,red] node[above,align=left] {$\ni \LmustVASSe \not\in$ \quad
  $\not\ni \LnotHDVASS \in$} (L0)

  (eLHD) edge[-,red,dashed] node[midway, above left,rotate=8] {$\ni\Lanblen\not\in$\quad\quad} (eLHD0)

  (L) edge node[below,align=left,xshift=-1.5cm,yshift=-0.1cm] { } (L0)
  (eL) edge (eL0)

  ;

  \node[] at (2.75,-2.9) {{\color{blue!50!black} $\not\ni \ $}$\Lanbn${\color{green!50!black} $\ \in$}};

\node[blue!50!black] at(-1.5,-2.8) {Coverability Semantics};
\node[green!50!black] at(7,-2.8) {Reachability Semantics};

\path[fill=blue!50!black,opacity=0.1]
({-1.5+3.5},-3.1)
rectangle
({-1.5-3.5},4.5);

\path[fill=green!50!black,opacity=0.1]
({7+3.5},-3.1)
rectangle
({7-3.5},4.5);

\end{tikzpicture}
 \end{adjustbox}
\begin{center}
    \tiny
\begin{tabularx}{\textwidth}{ |l|X|c|c| } 
 \hline
 Language    & Definition & Alphabet &Page \\
 \hline
 $\LnotDVASS$  
    & $ a^nb^{\le n} + a^*b^*c$
    & $\{a,b,c\}$
    & \pageref{lang:LnotDVASS}\\
  $\Lanbgen$
    & $a^nb^{\ge n}\#$
    & $\{a,b,\#\}$
    & \pageref{lang:Lanbgen}\\
 $\LnotHDVASS$  
    & $(a+b)^*a^nb^{\le n}$
    & $\{a,b\}$
    & \pageref{lang:LnotHDVASS}\\
  $\Lanblen$
    & $a^nb^{\le n}$
    & $\{a,b\}$
    & \pageref{lang:Lanblen}\\
  $\Lanbn$
    & $a^nb^n$
    & $\{a,b\}$
    & \pageref{lang:Lanbn}\\
  $\LmustVASSe$
    & $\mathit{bin}(n)\#0^{\le n}\# $, where $\mathit{bin}(n)$ is $n$ in binary.
    & $\{0,1,\#\}$
    & \pageref{lang:LmustVASSe}\\
  $\Lanblenbarrier$
    & $a^nb^{\le n}\#$
    & $\{a,b,\#\}$
    & \pageref{lang:Lanblenbarrier}\\
 \hline
\end{tabularx}
\end{center}
\caption{Comparison of expressive power of $\vass$ and $\hdvass$ language classes, with and without silent transitions, in reachability and coverability semantics.
A solid arrow $\lang{A}\xrightarrow{}\lang{B}$ indicates strict inclusion $\lang{A}\subsetneq \lang{B}$, with a separating language denoted on the edge. A red/dashed line indicates pair-wise incomparability, with the separating languages denoted. Dotted arrows indicate a special case.
}
\label{fig:comparison}
\end{figure}
In terms of the classes of languages they define,
history-deterministic VASSs
are strictly more expressive than deterministic ones, and in turn strictly subsumed by fully non-deterministic ones.
The following theorem states this formally. Its proof is split into \cref{l1notdvass,lemma:lastblocknotHD,lemma:rhvsrd,lemma:rhvsrn}.
\begin{theorem}
\label{thm:strictseparation}
	For all $k\geq 1$, we have $\CDvass[k] \subsetneq \CHvass[k] \subsetneq \CNvass[k]$ and $\RDvass[k] \subsetneq \RHvass[k] \subsetneq \RNvass[k]$.
\end{theorem}
\begin{lemma}
\label{l1notdvass}
\label{lang:LnotDVASS}
	$\LnotDVASS \eqdef a^nb^{\le n} + a^*b^*c \in \CHvass[1] \setminus \CDvass$.
\end{lemma}
\begin{proof}
$\LnotDVASS$ can be recognised by the \Ahvass[1] depicted in~\cref{fig:notDVvass}. Note that the \vass is HD: the only non-deterministic choice is whether to go to $q_2$ or $q_3$ on $b$, for which the resolver must always choose $q_2$ if available (if the counter is non-zero). The choice of resolver is unique as going to $q_3$ unnecessarily is not language maximal. 

For a contradiction, suppose $\LnotDVASS $ is accepted by a \Advass[k] with $n$ states. Since $w_{n+1} =a^{n+1}b^{n+1}\in \LnotDVASS$ the run is accepted. There exists $i<j \le n+1$ such that $a^{n+1}b^{i}$ is in state $q$ with counter vector $v\in\mathbb{N}^k$ and $a^{n+1}b^{j}$ is also in state $q$ with counter vector $v'\in\mathbb{N}^k$. Since $a^{n+1}b^{i}\in \LnotDVASS$, we have that state $q$ is accepting.

Suppose $v'- v \ge \vec{0}$, then $a^{n+1}b^{i+(j-i)n}\not\in \LnotDVASS$ is accepted. Therefore there exists a dimension such that $v'-v$ is negative. Hence for some $\ell$ we have $a^{n+1}b^{i+(j-i)\ell}$ is a dead run. Hence it cannot accept  $a^{n+1}b^{i+(j-i)\ell}c \in \LnotDVASS$.
\end{proof}
\begin{figure}[t]
\centering
\begin{subfigure}[t]{0.65\textwidth}
\includegraphics[width=0.9\textwidth]{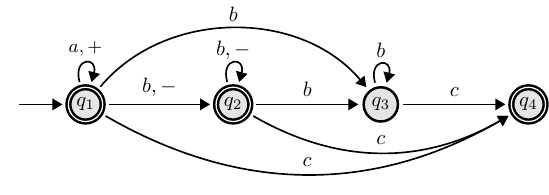}

\caption{A \Ahvass[1] recognising $\LnotDVASS$. }
\label{fig:notDVvass}

\end{subfigure}
\begin{subfigure}[t]{0.32\textwidth}
\centering
\includegraphics[width=\textwidth]{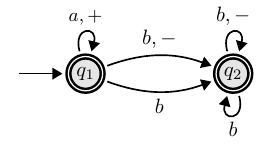}
\caption{A \ARhvass[1] recognising $\Lanbgen$.}
\label{fig:notD0vass}
\end{subfigure}
\caption{Transitions labelled with $+$ increment the counter by $1$, and those labelled by $-$ decrement the counter by $1$ and otherwise have no effect on the counter. }
\end{figure}

\begin{lemma}
\label{lemma:rhvsrd}
$\Lanbgen \eqdef a^nb^{\ge n } \in \RHvass[1] \setminus \RDvass$
\label{lang:Lanbgen}
\end{lemma}
\begin{proof}
$\Lanbgen$ is recognised by the \ARhvass depicted in~\cref{fig:notD0vass}. On $b$ the resolver can choose between decrementing the counter and no effect, the resolver will always choose to decrement whenever the counter is non-zero.

We have $\Lanbgen\not\in\RDvass$. Suppose a $\ARdvass$ with $n$ states exists, consider the run on the word $w_{n+1} = a^{n+1}b^{n+1}\in\Lanbgen$. There exists two prefixes of the run in which $a^{n+1}b^{i}$ and $a^{n+1}b^{j}$ revisit a state, and so the system cycles through states on extension of $a^{n+1}b^{i}$ with $b^*$. Thus, in order to accept $w_{n+1}b^i$ for all $i$ the automaton must visit only accepting states throughout the cycle. Since $a^{n+1}b^{i}\not\in \Lanbgen$ the counter must be non-zero, but zero at $u = a^{n+1}b^{i + (j-i)n}$ since $u \in L$, thus the effect of the cycle is decreasing on some counter, there must exist $k > n$ such that the run is dead on $a^{n+1}b^{i+(j-i)k}$. This is a contradiction as  $a^{n+1}b^{i+(j-i)k}\in\Lanbgen$.
\end{proof}

\begin{lemma}\label{lemma:lastblocknotHD}
\label{lang:LnotHDVASS}
	$\LnotHDVASS\eqdef \{a,b\}^*a^{n > 0}b^{\le n} \in\CNvass[1] \setminus \CHvass$. 
\end{lemma}
\begin{proof}
$\LnotHDVASS$ can be accepted a $\Anvass[1]$, which non-deterministically guesses the start of the last $a^*b^*$ block and accepts if there are fewer $b$'s than $a$'s.

We show that $\LnotHDVASS\not\in \CHvass$. Suppose for contradiction there is a $k$-$\hdvass$ with $|Q|$ states, $\norm{\delta}$ the largest effect on a counter in any transition and a resolver $r$.

	Consider a sequence of accepted words $w_\ell = w_{\ell-1} a^{m_{\ell}}b^{m_{\ell}}$, with $w_{0}$ the empty word, where $m_\ell$ is large enough so that there exist $ r_{\ell,1} < r_{\ell,2} \le m_{\ell}$, such that the run given by the resolver $r$ on  
	$w_{\ell-1} a^{r_{\ell,1}}$ has configuration $(q_\ell, v_\ell)$ and $w_{\ell-1} a^{r_{\ell,2}}$ has $(q_\ell, u_\ell)$, with $u_\ell \ge v_\ell$.
	In other words, whilst reading $a^{m_{\ell}}$, the run encounters a cycle on state $q_{\ell}$ which does not strictly decrease any counter value. This occurs due to Dickson's lemma and depends on $|Q|, \norm{\delta}, k$ and $m_{1},\dots,m_{\ell-1}$.
	This gives an inductive way to build words $w_{\ell}$ consisting of $\ell$ blocks of $a$s and $b$s such that each $a$-block visits a non-decreasing cycle.
	We consider the word $w_{n}$ for $n = 2^{k} + 1$ and the run $\rho$ on $w_n$ given by the resolver.

	Given a vector $v\in \N^k$, we define $support(v) = \{i\mid v_i \ne 0\}$. Since there are $n$ blocks of $a$ in $w_n$, each of which has a non-decreasing cycle $(q_\ell,u_\ell)$ and $(q_\ell,v_\ell)$, for $\ell\in \{1,\dots,n\}$. However, there are $2^k+1$ possible choices for  $support(u_\ell- v_\ell)$. Therefore, there exists $\ell< \ell'$ such that $support(u_\ell- v_\ell) = support(u_{\ell'}- v_{\ell'})$. In other words, there are two $a$-blocks which have a non-decreasing cycle such that the effect of the cycles have the same support.
	Let $R\in \N$ be such that $R(u_\ell - v_\ell) \ge u_{\ell'} -v_{\ell'}$, which exists since $support(u_{\ell}-v_{\ell})=support(u_{\ell'}-v_{\ell'})$ and 
	$u_{\ell}-v_{\ell}\ge\vec{0}$ and $u_{\ell'}-v_{\ell'}\ge\vec{0}$.
	
	Let $u$ be the word such that $w_{\ell'-1} = w_{\ell} u $, i.e, the part between the $\ell$th $b$-block and $\ell'$th $a$-block.
	Consider the word $w' =  w_{\ell-1} a^{m_\ell + R(r_{\ell,2} - r_{\ell,1})}b^{m_\ell}  u   a^{m_{\ell'} - (r_{\ell',2} -r_{\ell',1})}b^{m_{\ell'}}$.
	The word $w'$ is therefore obtained by adding $R(r_{\ell,2}-r_{\ell,1})$ many $a$'s in the $\ell$th $a$-block and removing $(r_{\ell',2}-r_{\ell',1})$ many $a$'s from the $\ell'$th $a$-block. Note that $w'\not\in \LnotHDVASS$, since the last block has more $b$'s than $a$'s. We will show that there is an accepting run on $w'$, by modifying the resolver run on $w_\ell'$.
	
	Let $\rho_{\ell'}$ be the run on $w_{\ell'}$ given by the resolver $r$. We consider the run $\rho'$ where we take the cycle between $(q_{\ell},v_{\ell})$ and $(q_{\ell},u_{\ell})$ an additional $R$ times in the $\ell$-th $a$-block, but removes the cycle between $(q_{\ell'},v_{\ell'})$ and $(q_{\ell'},u_{\ell'})$. We show that $\rho'$ is a run on $w'$. To see this, we must verify  that no counter drops below zero in $\rho'$.
	Note that the runs $\rho_{\ell'}$ and $\rho'$ are the same till the prefix $w_{\ell-1}a^{r_{\ell,2}}$ after which it reaches the configuration $(q_\ell,u_\ell)$. Then it does $R$ additional cycles which results in the configuration $(q_\ell,u_\ell+ R(u_\ell-v_\ell))$. From this point $\rho'$ follows the same sequence of transitions as $\rho_{\ell'}$ till it reads the prefix up to $w_{\ell'-1}a^{r_{\ell',1}}$ ending up in the configuration $(q_{\ell'},v_{\ell'}+ R(u_\ell-v_\ell))$. Since $v_{\ell'}+ R(u_\ell-v_\ell)\geq v_{\ell'} + (u_{\ell'}-v_{\ell'})=u_{\ell'}$, $\rho'$ can follow the suffix of the run $\rho_{\ell'}$ from $(q_{\ell'},u_{\ell'})$ on $a^{m_{\ell'}-r_{\ell',2}}b^{m_{\ell'}}$, which ends in the same state as $\rho_{\ell'}$ with a non-zero counter value.
	This is a contradiction as we get a accepting run on $w'\notin\LnotHDVASS$. We conclude that there is no $k{-}\hdvass$ that recognises the language $\LnotHDVASS$.
\end{proof}

\begin{figure}
\centering
\includegraphics[width=\textwidth]{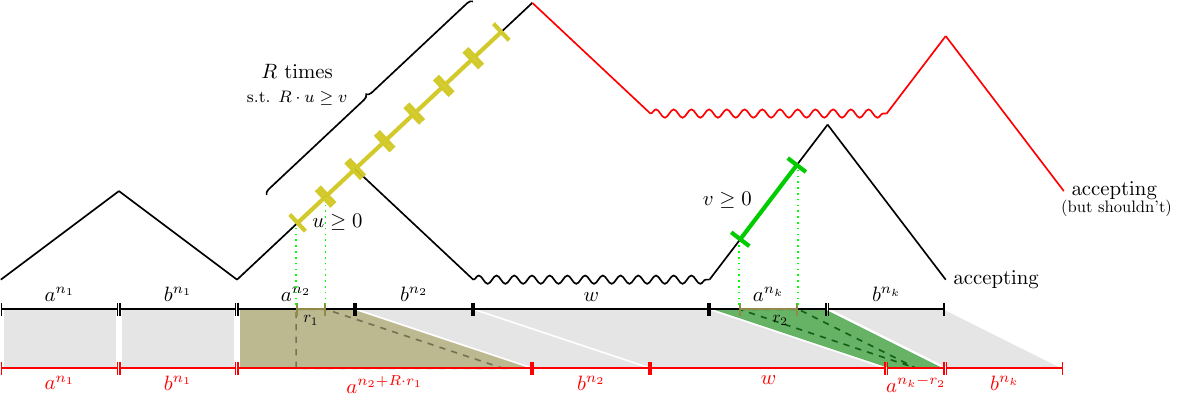}
\caption{Proof that $\LnotHDVASS \not\in \CHvass$ (\cref{lemma:lastblocknotHD}). For two cycles of lengths $r_1,r_2$ chosen in different $a^*$-blocks with effects $u,v \ge \vec{0}$ and $support(u) = support(v)$, repeating the first cycle and removing the second one
constructs an accepting run on a word $\notin\LnotHDVASS$.}
\label{fig:notHDpumping}
\end{figure}

\begin{lemma}
\label{lemma:rhvsrn}
$\Lanblen \eqdef a^nb^{\le n } \in \RNvass[1] \setminus \RHvass$
\label{lang:Lanblen}

\end{lemma}
\begin{proof}
In the non-deterministic case  reachability semantics can recognise 
$\Lanblen \in \RNvass[1] $: On $a$, non-deterministically choose either to increment by $1$ or not, guessing ahead of time how many $b$'s will be seen. On $b$, the machine moves to a new state and counts down, preventing more $b$'s than the guessed number.

However $\Lanblen$ cannot be recognised with history-determinism. To see this, observe that since $a^n \in \Lanblen$ all the counters must be zero after reading $a^n$, then, for $n$ larger than the number of states, the machine cannot distinguish $a^nb^n\in \Lanblen$ and 
$a^nb^{n+1}\not\in \Lanblen$. 
\end{proof}

\subsection{Silent transitions} %
\label{sub:expr:eps}

First observe that
\label{lang:Lanbn}
$\Lanbn \eqdef a^nb^n$
can be recognised with reachability semantics (even $\RDvass$), but cannot be recognised under coverability semantics (even \CNvassE). On the other hand  $\Lanblen = a^nb^{\le n }$ can be recognised by coverability semantics (even $\CDvass$), but cannot be recognised by $\RHvassE$, thus together $\Lanblen$ and $\Lanbn$ show pairwise incomparability between reachability and coverability semantics for deterministic and history-deterministic systems. 
However, if the languages have an end marker then coverability acceptance can be turned into reachability acceptance (with $\epsilon$-transitions) as $\epsilon$-transitions can be used to take the counters to zero at the end marker.

The separation between $\CNvass$ and $\CNvassE$ is due to~\cite{G1978} for which 
\label{lang:LmustVASSe}
$\LmustVASSe \eqdef \mathit{bin}(n)\#0^{\le n}\# \in \CNvassE\setminus \CNvass$, where $bin(n)$ is the binary representation of $n\in\N, n > 0$ in $1\{0,1\}^*$. 
This language cannot be recognised without $\epsilon$ transitions (see \cref{app:epsrequired} for details).
We observe that the same language separates \CHvass and \CHvassE, as the $2$-VASS of \cite{G1978} recognising $\LmustVASSe$ is in fact history-deterministic. 
However, in dimension $1$, the two classes collapse: %

\begin{restatable}{lemma}{hdoneishdonee}
\label{lemma:hd1ishd1e}
$\CHvass[1] = \CHvassE[1]$.
\end{restatable}

While in coverability semantics, the presence of $\epsilon$-transitions separates languages recognised by $\Ahvass[k]$ and $\AhvassE[k]$ only for dimensions $k\ge 2$, in reachability semantics the separation occurs already in dimension $1$:
\label{lang:Lanblenbarrier}
$\Lanblenbarrier \eqdef a^nb^{\le n}\#$  is in $\RHvassE$ but not in $\RHvass$.

\subsection{Comparison with Finitely Sequential VASS}
\label{sub:expr:fs}
Recall that finitely sequential \vass are the union of finitely many $\dvass$. In \cref{lem:unionclosed}
we show that language of a finite union of history-deterministic VASS is also history-deterministic. In particular, the deterministic VASSs comprising the finitely sequential VASS are themselves history-deterministic, so any finitely sequential VASS has an equivalent history-deterministic VASS recognising the same language. On the other hand, we show that history-deterministic VASS with coverability acceptance are strictly more powerful:

\begin{lemma}
There exists a language in $\CHvass[1]$ that  is not finitely sequential.
\end{lemma} 

\begin{figure}
\centering
\includegraphics{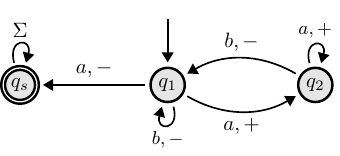}
\caption{A $1$-\hdvass automaton with language $\LnotFSVASS = \Lang{q_1,0}$ that is not finitely sequential. The automaton reads blocks of $a$'s followed by blocks of $b$'s. If some block of $a$'s is followed by fewer $b$'s then the automaton can read anything after the next $a$. If every block is followed by the same number of $a$'s and $b$'s then it must read another block of the form $a^{n}b^{n}$ or $a^{n}b^{< n}$.
The language is thus $\LnotFSVASS=\bigcup_{k=0}^{\infty} a^{n_0}b^{n_0}\dots a^{n_{k-1}}b^{n_{k-1}}a^{n_{k}} b^{<n_{k}}a\Sigma^*$.
}
\label{fig:notfs}
\end{figure}
\begin{proof}
\label{lang:LnotFSVASS}
Consider the language $\LnotFSVASS \eqdef \Lang{q_1,0}$ 
of the VASS depicted in \cref{fig:notfs}. 
Observe that it is history-deterministic: when reading $a$ at state $q_1$, the resolver goes to $q_s$ if possible. This choice is language-maximal and there is no other non-determinism to resolve.

We show the language is not finitely sequential. Suppose for contradiction the language is accepted by a finitely sequential VASS that is the union of $k$ many \Advass, each with at most $m$ states.
We consider the word $a^{m+1}b^{m+1}$, which can be extended into an accepting word in $\LnotFSVASS$, and thus is alive in some \Advass. For \Advass in which the run is still alive, reading this word goes through a cycle in the run while reading $a^{m+1}$ and similarly also whilst reading $b^{m+1}$. 

Let $c_1,\dots,c_k$ be the lengths of these cycles while reading $a$'s in each \dvass respectively, $d_1,\dots,d_k$ be the lengths of the cycles reading $b$'s, and fix
$C = \prod_{i\le k} c_i$
and $D= \prod_{i\le k} d$. 
Observe that, for every $x$, for each $\dvass$ the same state is reached after reading $a^{m+1+xC}$. Similarly, for any $y$, the same state is reached after reading $a^{m+1+xC}b^{m+1+yD}$. 
In particular, fix words $w = a^{m+1+CD}b^{m+1+CD}$ and the words $u = a^{m+1+CD}b^{m+1+(C-1)D}$.

Observe that after reading $ua$, the system in  \cref{fig:notfs} can be in state $q_s$ and therefore, any extension of $ua$ is accepted.
However, the automaton of \cref{fig:notfs} can only reach state $q_2$ on $wa$ and so, for any $z \in\N, i \ge 1$, $wa^{z}b^{z+i} \not\in \LnotFSVASS$.
Consider this word for $z = m+1$. Since there is a cycle somewhere while reading $b^{z}$, then when reading more $b$'s the automaton visits only states on that cycle. Since $wa^{z}b^{z+i} \not\in \LnotFSVASS$ for $i\ge1$ either every state on the cycle is non-accepting, or the cycle has a negative effect on at least one counter and therefore becomes unavailable for large enough $i$.

Recall, in $M$ both $wa$ and $ua$ are in the same control location in each constituent $\dvass$, and thus for any $v\in\Sigma^*$ we have
$wav$ and $uav$ reach the same control locations (or possibly the run is dead). However, for every $z,i$, there is some \dvass in which the word $ua^{z}b^{z+i}$ is accepting. However, we have argued that for every \dvass, for sufficiently large $i$, the run on $wa^{z}b^{z+i}$ is stuck in a rejecting cycle, or a cycle in which the counter is decreasing. Thus for sufficiently large $i$, in every \dvass, either the run on $ua^{z}b^{z+i}$ is also dead or in a rejecting cycle, which contradicts $ua^{z}b^{z+i}\in \LnotFSVASS$.
\end{proof}

\section{Closure Properties}
\label{sec:closure}

We take a look at closure properties of the classes $\CHvass$ and $\RHvass$
recognised by history-deterministic VASSs in coverability and reachability semantics, respectively.

Union closure (of \CHvass and \RHvass) and closure under intersection (for \CHvass) can be shown using a straightforward product construction at the cost of increasing the dimension.

\begin{restatable}{lemma}{lemunion}
    \label{lem:unionint}
    \label{lem:unionclosed}
 	Let $L\in\CHvass[k]$ and $L'\in \CHvass[k']$. Then $L\cup L'\in \CHvass[(k+k')]$ and $L\cap L'\in \CHvass[(k+k')]$.
 	Let $L\in\RHvass[k]$ and $L'\in \RHvass[k']$. Then $L\cap L'\in \RHvass[(k+k')]$. 
\end{restatable}
A na\"ive product of the two systems recognising $L$ and $L'$ does \emph{not} work for showing the union closure of $\RHvass$
because here, acceptance requires all counters to be zero
even for inputs that are only in one of the two languages (note the absence of $\eps$-transitions).
Indeed, $\RHvass$ are not closed under union, as witnessed by $\Lnotrunion\eqdef a^nb^n\cup a^nb^{2n}$ not being in $\RHvass$
(see \cref{app:closure}). \label{lang:lnotrunion}

Taking a direct product yields a $\Ahvass$ that may not be optimal in terms of the number of counters
and in general, increasing the dimension is not avoidable.
For instance,
the languages $\Lnotcunion\eqdef a^nb^{\le n}c^*\cup a^nb^*c^{\le n}$ and $\Lnotcint \eqdef a^nb^{\le n}c^*\cap a^nb^*c^{\le n}$ 
are not in \CHvass[1], while the individual component languages are.
Similarly,
the language $\Lnotrint \eqdef a^nb^nc^*\cap a^nb^*c^n=a^nb^nc^n$ 
witnesses non-closure of \RHvass[1] under intersection.
\label{lang:lnotrint}
\label{lang:lnotcunion}
\label{lang:lnotcint}

The theorems below summarise our findings regarding closure properties of history-deterministic classes. Full proofs are in \cref{app:closure}.
 
 \begin{restatable}{theorem}{thmcovcl}\label{thm:coverclosure}
	$\CHvass$ is closed under union, intersection and inverse homomorphisms.\\
	It is not closed under complementation, concatenation, homomorphisms, iteration, nor commutative closure.
\end{restatable}

\vspace{-0.25cm}
\begin{restatable}{theorem}{thmreachcl}\label{thm:reachclosure}
$\RHvass$ is closed under intersection and inverse homomorphisms.\\
It is not closed under union, complementation, concatenation, homomorphisms, iteration, nor commutative closure.
\end{restatable}

\section{Decision Problems}
\label{sec:decision}
\newcommand{\CorrectUpto}[1]{\mathsf{Correct}_{#1}}
\newcommand{\IncorrectAt}[1]{\mathsf{Incorrect}_{#1}}

In this section we consider decision problems related to history-determinism:
checking if a given \Anvass is history-deterministic,
HD definability (as well as regularity) of its recognised language,
and language inclusion between HD VASSs.

Prakash and Thejaswini \cite{PT2022} showed that in dimension $1$
(and for coverability semantics),
checking HDness and inclusion is decidable in \PSPACE\
by reduction to simulation preorder \cite{HLMT2016}.
This can be generalised slightly as follows.

\begin{theorem}
Language inclusion $\Lang{\?B}\subseteq\Lang{\?A}$
is decidable for any \Ahvass[1] $\?A$ and for any \Anvass $\?B$.
\end{theorem}
\begin{proof}
By Theorem~19 in \cite{PT2022}, for any \Ahvass[1], one can effectively construct
a language equivalent deterministic one-counter automaton (DOCA; a \Avass[1] with zero-testing transitions).
DOCA can be complemented \cite{VP1975} and so the inclusion question is equivalent to the 
emptiness (reachability) of $\overline{\?A}\x\?B$, a VASS with one zero-testable counter, which
is decidable \cite{R2008}. Note this result is independent of the number of counters of $B$.
\end{proof}

We continue to show that in higher dimensions, these questions are undecidable.
Throughout this section, when we show undecidability for $\epsilon$-free VASS, the result naturally also applies for superclass with silent transitions.
Our constructions proving this are similar, yet require subtle differences,
and are all based on weakly simulating two-counter machines \cite{M1967}.
Let us recall these in a suitable syntax first.

\begin{definition}
A two-counter Minsky machine (2CM) $M=(\states, q_0,q_h, \transitions)$
consists of a finite set of states $\states$, including a distinguished starting and final state $q_0, q_h$, respectively,
as well as a finite set of transitions $\transitions \subseteq \states\x\CMOPS\x\states$, 
where $\CMOPS=\{\inc_1,\inc_2,\dec_1,\dec_2,\ztest_1,\ztest_2\}$
are the operations on the counters\footnote{Readers may be more familiar with an instruction of the form $\textsf{if } C_i = 0 \textsf{ goto } q_\ell \textsf{ else goto } q_k $, this can be simulated by a $\ztest_i$ to $q_\ell$ and a decrement followed by an increment to $q_k$.}.

A configuration of $M$ is an element of $\states\x\N^2$, comprising the current state and the value of the two counters. 
For every state $q$ either:
\begin{enumerate} %
    \item There is only one transition of the form $(q, \inc_i, q')$.
        This allows to move from state $q$ to $q'$, increment counter $i$ by one and leaves the other counter untouched; or
    \item There are exactly two transitions from $q$, of the form $(q, \ztest_i,q')$ and $(q, \dec_i,q'')$.
        The former allows to move to $q'$ without changing the counters, but only if counter $i$ has value $0$.
        The latter allows to move from $q$ to $q''$ and decrease counter $i$, and leaves the other counter unchanged.
\end{enumerate}
\end{definition}

Notice that from any configuration there is exactly one possible successor configuration.
We can therefore speak of \emph{the} run of $M$, and its sequence of counter operations, from the initial configuration $(q_0,0,0)$.
We say that $M$ \emph{terminates} if its run visits the final state $q_h$.
W.l.o.g., we can assume that both counters have value $0$ whenever $M$ terminates.

Deciding whether a given 2CM terminates is undecidable \cite{M1967}.
An easy consequence, and the basis for our construction for regularity,
is the undecidability of checking finiteness of the reachability set for a given 2CM.

\begin{restatable}{lemma}{cmfinite}
    \label{lem:2cm-finiteness}
    It is undecidable to check, for given 2CM $M$, 
    if its run visits infinitely many different configurations.
\end{restatable}

\subsection{Checking \hdness and Inclusion}
\label{ssec:hdness-sys}
We focus on the questions of whether a given $\Avass$ is history-deterministic,
and whether language inclusion holds for two languages given by $\Ahvass$.
For languages of finite words these two decision problems are intrinsically linked (see \cref{app:lang-max}).

\begin{figure}[t]
    \centering
    \includegraphics[width=0.8\textwidth]{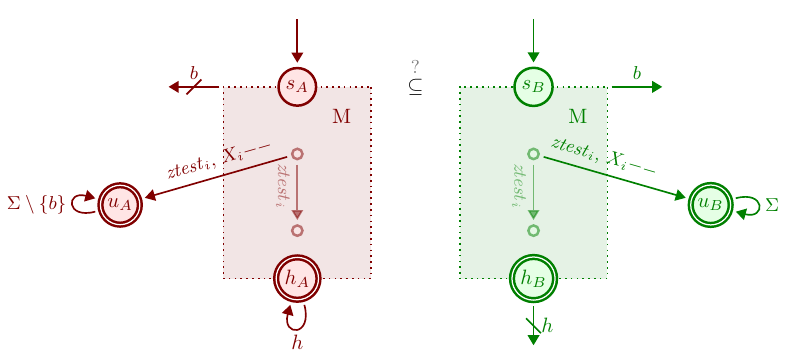}
    \caption{The $2$-VASSs $A$ (in red) and $B$ (in green) both include a copy of, and weakly simulate, a given 2CM $M$.
        For any zero-testing operation $\ztest_i$ in $M$ both can go to a sink state
        if counter $i$ is in fact non-zero, reading the letter $\ztest_i$ and decreasing the VASS counter $i$, as indicated by the effect vector $\decop{i}$.
        The extra letter $b$ ensures that $\Lang{B}\not\subseteq\Lang{A}$;
        Only $A$ can accept words that consist of valid sequences of 2CM operations and that end in the letter $h$.
    }
    \label{fig:2-hd-undec}
\end{figure}

\begin{lemma}\label{lem:2-2-vass}    
For a given 2CM $M$ one can construct two history-deterministic $2$-VASSs
with initial states $s_A$ and $s_B$, respectively, so that
$\Lang{s_A,0}\subseteq\Lang{s_B,0}$ if, and only if,
the unique valid run of $M$ never reaches a halting state.
\end{lemma}
\begin{proof}
    Suppose we are given 2CM $M$
    with designated initial and halting states $s$ and $h$, respectively,
    and let $\CMOPS$ denote the set of counter operations.
    W.l.o.g., there is exactly one valid sequence of counter operations that is either infinite or finite.
    We define two $\Avass[2]$s $A$ and $B$
    over the alphabet $\Sigma=\CMOPS\uplus\{b,h\}$.
    These are just copies of, and just weakly simulate the machine $M$:
    For every state $q$ of $M$, there are states $q_A$ and $q_B$;
    For every transition $q\step{\op}q'$ of $M$,
    there are corresponding edges
    $q_A\step{\op}q'_A$ and $q_B\step{\op}q'_B$
    that read the letter $\op$ and manipulates the counter accordingly:
    if $\op=\inc_i$ (or $\dec_i$) then counter $i$ is incremented (or decremented, respectively).
    If $\op=\ztest_i$ then counter $i$ remains as is.
    The only accepting states so far are $h_A$ and $h_B$, corresponding to the designated halting state of $M$.

    Additionally, for every zero-testing transition $q\step{\ztest_i}q'$ in $M$,
    both $A$ and $B$ have a transition from state $q$ that decreases counter $i$
    and goes to a new, accepting, sink state $u$
    with language $\supseteq (\CMOPS\cup\{h\})^*$.
    This way, both systems will accept any word that prescribes a run of $M$
    that contains a ``counter cheat'', meaning that the word contains operation $\ztest_i$
    but the run of $M$ so far ends in a configuration where counter $i$ is not zero.

    We now modify the systems $A$ and $B$ so that they differ in two ways:
    \begin{enumerate}
        \item 
            the halting state $h_A$ 
            of $A$ admits a $h$-labelled step (to itself)
            but $s_B$ does not.
        \item All states in $B$ have $b$-labelled steps (to the accepting sink $u_B$) but none of $A$s states do.
    \end{enumerate}
    See \cref{fig:2-hd-undec} for a depiction of the constructed $2$-VASS.

    \medskip
    Notice that $\Lang{B}\not\subseteq \Lang{A}$ by design, because no word containing the letter $b$ can be accepted by $A$.
    Notice that both $A$ and $B$ are indeed history-deterministic:
    the only choices to be resolved are upon reading a zero-testing letter $\ztest_i$ from a configuration where the corresponding counter $i$ is not zero.
    In any such case, moving to the sink is language maximal.

    It remains to argue that $\Lang{s_A} \not\subseteq \Lang{s_B}$ if, and only if, $M$ has a finite run from an initial configuration to its final state.
    Indeed, if $M$ terminates via a sequence $\rho=e_0e_1,\ldots e_k$, then $\Lang{A}$ contains the word $\rho \cdot h$.
    Since this run does not contain ``cheats'' nor letters $b$, the system $B$
    cannot possibly reach the winning sink $u_B$ and therefore not accept.
    Conversely, if $M$ does not terminate,
    then any word $\rho\in \CMOPS^*\cdot h$ accepted by $A$ must prescribe a
    run of $M$ that contains a cheat. Say $\rho=\rho_1 \cdot \ztest_i \cdot \rho_1 \cdot h$. But then, $B$ will be able to reach the sink $u_B$ after reading the prefix $\rho_1\cdot \ztest_i$ and thus accept.
\end{proof}
The construction in the previous lemma works both in coverability and reachability semantics (note that we assume that a 2CM terminates with counters at $0$). The next two theorems are direct consequences and again hold for coverability and reachability semantics.
\begin{theorem}
    \label{thm:2-HD-inlusion}
    Checking language inclusion is undecidable for $2$-HD VASSs.
\end{theorem}
\begin{theorem}
    \label{thm:2-HDness}
    It is undecidable to check if a given $2$-VASS is history-deterministic.
\end{theorem}
\begin{proof}
    By reduction from 2CM termination:
    Construct the two systems $A$ and $B$ as given by \cref{lem:2-2-vass}
    and add one new initial state $s$ that, upon reading some letter $b$
    can move to the initial state $s_A$ of $A$ or $s_B$ of $B$.
    By \cref{prop:hd=language-maximal}, the so-constructed system is HD iff
    $\Lang{s_A}\subseteq\Lang{s_B}$, which is true iff $M$ does not terminate.
\end{proof}

\subsection{Checking \hdness of VASS Languages}
\label{ssec:hdness-lang}
We turn to showing undecidability of \emph{language} history-determinism, i.e., the question if for a given VASS there exists an equivalent history-deterministic VASS.
We start with the more interesting and involved case, for the coverability semantics
(\cref{theorem:languagehdnessundecidable})
and present an easier construction for reachability (\cref{thm:languagehdness-reach}) afterwards.

We give a proof by reduction from the 2CM halting problem,
combining the constructions to show the non-HDness of $\LnotHDVASS=(a,b)^*a^nb^{\le n}$, (\cref{lemma:lastblocknotHD})
and the proof of \cite{JM1995} that checking regularity for \Anvass languages is undecidable.

\begin{restatable}{theorem}{thmlanguagehdnesscover}
    \label{theorem:languagehdnessundecidable}
    \label{thm:languagehdness-cover}
    It is undecidable to check if $\Lang{\?A} \in \CHvass$ holds for a given \Anvass $\?A$.
\end{restatable}
\begin{proof}

    \newcommand{\sinkstate}{\mathit{sink}}
    By reduction from the $2$CM halting problem.
    For a given $2$CM $M$ with states $\states[M]$ and counter operations $\CMOPS = \{\inc_1,\inc_2,\dec_1,\dec_2,\ztest_1,\ztest_2\}$
    we construct a \Avass[3] $\?A=\vasstuple$
    so that $\Lang{\?A}$ is history-deterministic iff the faithful run of $M$ is finite.

    We refer to the three counters as $X_1,X_2,X_3$
    and write $\decop{i}$ and $\incop{i}$
    for the effects of (VASS) transitions that decrement/increment counter $i$ only.
    
    \medskip
    \textbf{The construction.}
    $\?A$ uses the alphabet $\Sigma = \CMOPS \cup \{a,b\}$, consisting of counter operations of $M$ and two fresh symbols.
    The control states of $\?A$ mimic those of $M$, except that in between any simulated step of $M$,
    $\?A$ can read a word in $a^+b^+$:
    For every state $q\in \states[M]$ we introduce states $q_{in}$, $q_{out}$ and $q_{step}$.
    In addition, we add three other states $\sinkstate,r_1,r_2$.
    We make $\sinkstate$ universal by adding self-loops $(s,a,\vec{0},s)$ for every letter $a\in\Sigma$.
    First we consider the simulation of $M$.

    For every step $q\step{\op}p$ of $M$,
    $\?A$ has a transition
    $t=(q_{out},\op,e,p_{in})$
    from $q_{out}$ to $p_{in}$
    that reads the letter $\tfunction{t}=\op$ and manipulates the counter accordingly:
    if $\op=\inc_i$ 
    then $e=\incop{i}$;
    if $\op=\dec_i$ then $e=\decop{i}$;
    if $\op=\ztest_i$ then $e=\vec{0}$.
    In addition, for zero-testing steps 
    $q\step\ztest_i p$, $\?A$
    in $M$, $\?A$ contains
    a decreasing transition 
    $t=(q_{out},\ztest_i,\decop{i},\sinkstate)$ 
    to the universal sink state.
    From a state $q_{in}$. There are two possible continuations:%
    \begin{enumerate}
        \item Reading a word in $a^+b^+$ and moving to $q_{out}$,
            via transitions $q_{in} \step{a,\vec{0}} q_{step}$, $q_{step}\step{a,\vec{0}} q_{step}$, $q_{step}\step{b,\vec{0}} q_{out}$ and $q_{out}\step{b,\vec{0}} q_{out}$. 
        \item 
            Reading a word in $a^{n}b^{\le n}$ and stopping.
            For this, there are transitions
            $q_{in}\step{a, \incop{3}} r_1$,
            $r_{1}\step{a,\incop{3}} r_1$,
            $r_{1}\step{b,\decop{3}} r_2$ and
            $r_{2}\step{b,\decop{3}} r_2$. 
    \end{enumerate}
    The accepting states of $\?A$ are $F=\{r_2,\sinkstate\}$.
    Its initial state is $\initialstate = q_{out}$, where $q\in\states[M]$ is the initial state of $M$.

    \medskip
    \textbf{The recognised language}
    of the constructed \Avass[3] $\?A$ contains sequences of instructions of $M$ interspersed with blocks of the form $a^+b^+$.
    Let's call a sequence $\op_1\op_2\dots\op_k\in \CMOPS^*$ of operations in $M$ \emph{faithful} if for all $i\le k$, $\op_i$ is the $i$th instruction in the run of $M$ from its initial configuration $(q,0,0)$. 
    Clearly, for any $k$ less or equal to the length of the run of $M$, there is a unique faithful sequence 
    $\rho_k$ %
    of length $k$.
    Define $\CorrectUpto{k} \eqdef \op_1(a^+b^+)\op_2(a^+b^+)\op_3 \dots (a^+b^+)\op_k$
    where $\op_1\op_2\dots\op_k = \rho_k$.
    Let
    $\IncorrectAt{k}\subseteq \Sigma^*$ contain exactly all words $w\op\in\Sigma^*\setminus \CorrectUpto{k}$
    where $w\in\CorrectUpto{k-1}$ and $\op\in\{\ztest_1,\ztest_2\}$. That is, words whose projection into the operations of $M$ is faithful
    up to step $k-1$ but that contain an incorrect zero-test at step $k$.

    Observe that if the faithful sequence of length $k$ takes $M$ to $(q,C_1,C_2)$
    then $\?A$ can read any word in $\CorrectUpto{k}$
    and every run on such a word leads to the configuration $(q_{in},C_1,C_2,0)$.
    Such a run of $\?A$ can be extended in two ways to reach an accepting state. Either by reading a word in $a^nb^{\le n}$ to reach $r_2$,
    or by continuing on the run of $M$ and eventually erroneously reading a  $\ztest_i$ to reach $\sinkstate$.
    We can therefore write the language of $\?A$ as
    \[
        \Lang{\?A}
        \quad=\quad 
        \bigcup_{k\ge 0} \CorrectUpto{k} \cdot (a^n b^{\le n})
        \quad
        \cup
        \quad
        \bigcup_{k\ge 0} \IncorrectAt{k}\cdot\Sigma^*
    \]
    \textbf{HDness.}
    We show that if $M$ terminates, meaning its run has some length $k\in\N$, then $\Lang{\?A}$ is history-deterministic.
    Observe that for every $0\le i\le k$, both languages $\CorrectUpto{i}$
    and $\IncorrectAt{i}$ are regular. 
    We can concatenate a DFA recognising the former with a \Ahvass[1] for $a^nb^{\le n}$
    to construct an \Ahvass[1] recognising $\CorrectUpto{i} \cdot (a^n b^{\le n})$.
    Observe that $\IncorrectAt{i}$, $i> k+1$ is empty.
    Now, 
    $\Lang{\?A}$ is the finite union of $k$ many $\Ahvass[1]$ languages (and a regular language)
    and therefore %
    recognisable by a $k$-dimensional \Ahvass.

    It remains to show that if the run of $M$ is infinite, then $\Lang{\?A}$ is not in $\CHvass[k]$, for any $k$. 
     Our proof mirrors the proof of \cref{lemma:lastblocknotHD}, except that we interleave
    $\{a,b\}$-blocks with the faithful operations of $M$.
    Suppose towards a contradiction that there exists a \Ahvass[k] $\?B$ with states $\states[\?B]$
    and let $\rho = \op_1\op_2,\dots \in \CMOPS^\omega$
    denote the infinite run of $M$. That is, every length-$i$ prefix $\rho_i$ is faithful.
    Consider a sequence $(w_n)_{n\ge0}$ of words in $\Lang{B}$ such that $w_0=\eps$ and otherwise %
    $w_\ell = w_{\ell-1} \op_{\ell} a^{m_{\ell}}b^{m_{\ell}}$ with $m_\ell$ large enough so that 
    the resolved run on $w_\ell$ contains a non-decreasing cycle while reading the last $a$-block.
    Say,
    \[
    (\initialstate,\vec{0})
    \step{w_{\ell-1}\op_\ell a^{r_{\ell,1}}}
    (q_\ell, u_\ell)
    \step{a^{r_{\ell,2}}}
    (q_\ell, v_\ell)
    \]
    with $u_\ell\le v_\ell$.
    This is well-defined by Dickson's Lemma.

    Setting $n = \abs{\states[\?B]}2^{k} + 1$
    is sufficiently high so that there must be $\ell<\ell'$
    with $q_\ell = q_{\ell'}$ and $\support{u_\ell- v_\ell} = \support{u_{\ell'}- v_{\ell'}}$.
    Take $R$ be such that $R(u_\ell - v_\ell) \ge (u_{\ell'} -v_{\ell'})$
    and let $u$ be the word such that $w_{\ell'-1} = w_{\ell} u$.
    Now consider the word
    \[
    w' = w_{\ell-1} \op_{\ell} a^{m_\ell + R(r_{2,\ell})}b^{m_\ell}  u  \op_{\ell'} a^{m_{\ell'} - r_{2,\ell'}}b^{m_{\ell'}}
    \]
    that results from $w_n$ by removing one iteration of the loop in block $\ell'$
    and making up for it by inserting $R$ iterations of the loop in block $\ell$.
    Notice that $w'$ is accepted by the run
    that follows the resolved run on $w_n$ and repeats the designated loops on the extra letters.
    However, $w'\notin\Lang{\?A}$ because its last $\{a,b\}$-block
    contains more $b$'s than $a$'s.
\qedhere
\end{proof}

Notice that if the given 2CM terminates then our construction produces
a history-deterministic VASS where the number of counters corresponds to the length of the terminating run.
Therefore it remains open whether the language $k$-$\hdness$ problem is decidable, which ask whether there is an equivalent $k$-HDVASS for the given language.

The analogous statement for reachability
is simpler to prove;
by adapting the construction for the regularity problem of Parikh-automata~\cite{EGJLZ22}, 
we reduce from the universality problem, which is undecidable in reachability semantics~\cite[Theorem 10]{VVN1981}. 
\begin{restatable}{theorem}{thmlanguagehdnessreach}
    \label{thm:languagehdness-reach}
    It is undecidable to check if $\Lang{\?A} \in \RHvass$ holds for a given \Anvass $\?A$.
\end{restatable}

\begin{proof}
    We reduce from the undecidable universality problem for \vass languages in reachability semantics
    \cite[Theorem 10]{VVN1981}. The construction is the same as
    for the regularity problem of Parikh-automata, recently presented in \cite{EGJLZ22}.
    For an alphabet $\Sigma$ 
    let $\Sigma_\$=\Sigma\uplus\{\$\}$ for some fresh symbol $\$\notin\Sigma$.
    For two words $u,v$ let $u\otimes v$ be the word 
    $w=(a_1,b_1)(a_2,b_2)\ldots (a_k,b_k)$ so that either 
    $u = a_1a_2\ldots a_k$ and $b_1b_2\ldots b_k\in v\$^*$
    or
    $v= b_1b_2\ldots b_k$ and $a_1a_2\ldots a_k\in u\$^*$.

    For two languages $L,L'\subseteq\Sigma^*$ define their \emph{cross-union}
    $L\ovee L\subseteq (\Sigma_\$^2)^*$ to be the langugae of words $u\otimes v$ such that $u\in L$ or $v\in L'$. That is, for any word $w \in L \ovee L$, either the projection into the first components is 
    $\pi_1(w)\in L$ or that into the second components $\pi_2(w)\in L'$.

    Recall the language $\Lanblen=a^nb^{\le n} \in\RNvass\setminus \RHvass$, which is not HD recognisable.
    To show our claim, let $L$ be some given \Anvass language and 
    consider the language $\Lcross \eqdef \$ \cdot (L\ovee \emptyset) \cup \$\cdot (\emptyset\ovee \Lanblen)$.
    This is clearly in \RNvass. %
    Now, if $L=\Sigma^*$ is universal then $L\ovee\emptyset$ is universal over $\Sigma_{\$}^2$ and so $\Lcross = \$(\Sigma_{\$}^2)^* \in\RHvass$ (even, regular).
    If conversely, suppose $L$ is not universal as witnessed by $w\notin L$,
    then $\Lcross$ cannot be recognised by any $\ARhvass$ for the same reason as $a^nb^{\ge n} \not\in \RHvass$:
    suppose it is accepted by some $\ARhvass[k]$ run on $n$ states and 
    consider run of the resolver on the word $u= \$(w  \otimes a^{|w|+n+1})\in \Lcross$, thus must end with counter $\vec{0}$. The extension of $u$ by $(\$,b)^{n+1}$ is also accepting, it must remain at $\vec{0}$ and cycle on accepting states. Hence $u(\$,b)^{|w|+n+1}\in \Lcross$ cannot be distinguished from $u(\$,b)^{|w|+n+2}\not\in \Lcross$.
\end{proof}

\subsection{Regularity}
\label{ssec:regularity}
We turn to the decision problem of whether a given \Avass recognises a regular language.
This regularity question is undecidable for general \Anvass \cite{JM1995}.
It again turns out that for history-deterministic VASSs, the decidability status of regularity depends
on the dimension.
For \Ahvass[1], one can effectively construct a language equivalent DOCA \cite{PT2022}, for which checking regularity remains decidable \cite{BGJ2014,VP1975}.

\begin{theorem}
    \label{thm:regularity1}
    Given a \Ahvass[1] $\?A$, checking if $\Lang{\?A}$ is regular is decidable in \EXPSPACE.
\end{theorem}

Although checking regularity of DOCA is \NL-complete,
the added complexity here is due to the doubly exponentially large DOCA produced in the reduction.
Since $\AhvassE[1]$ can be transformed into $\Ahvass[1]$ by \cref{lemma:hd1ishd1e}, the theorem also holds for $\AhvassE[1]$.
We now show undecidability already for dimension $2$.

\begin{theorem}
    \label{thm:regularity2}
    Given a \Ahvass[2] $\?{A}$, it is undecidable if $\Lang{\?A}$ is regular. 
\end{theorem}

\newcommand{\corr}[1]{\mathit{correct}_{#1}}
\newcommand{\incorr}[1]{\mathit{incorrect}_{#1}}
\begin{proof}
    By reduction from the finiteness problem for 2CM (\cref{lem:2cm-finiteness}).
For a given 2CM $M$ we construct a \Ahvass[2] whose language will be regular iff $M$'s run visits only finitely many configurations.
We make the argument for coverability semantics first.

Let $\rho=\op_1\op_2\ldots$ be the faithful run of $M$ and $\abs{\rho}\in\N\cup\{\infty\}$ for its length.
Write $\corr{k}$ for its length-$k$ prefixes and let $x_k$ %
be $1$ plus the sum of both counter-values in the configuration $M$ reaches after reading $\corr{k}$.
Further, wherever $\corr{k}=\corr{k-1}\dec_i$, define $\incorr{k}$ as $\corr{k-1}\ztest_i$.

Consider the language $L=G\uplus B$ over the alphabet $\Sigma=\CMOPS\uplus\{a\}$,
\[
    G \eqdef \bigcup_{k\ge 0}
    \left(
        \rho_k\cdot a^{\le x_k}
    \right)
    \qquad
    \text{and}
    \qquad
    B \eqdef \bigcup_{k\ge 0}
    \left(
    \incorr{k}\cdot\Sigma^*
    \right)
\]
$G$ consists of words that 
describe some length-$k$ prefix of $M$'s run followed by $x_k$ or fewer symbols $a$;
$B$ contains all words describing the run of $M$ up to length-$k$, followed by an incorrect zero-test, and then anything.

We claim that this language $L$ is recognised by a \Ahvass[2].
To see this, again build a \vass that weakly simulates $M$ as done before, for example in the proof of \cref{theorem:languagehdnessundecidable}.
This will simulate increment and decrement operations faithfully,
reading letters $\inc_i$ or $\dec_i$, respectively.
For any step $q\step{\ztest_i}q'$ in $M$, 
the \vass $\?A$ will have a transition $(q,\ztest_i,\vec{0},q')$
as well as one that reads $\ztest_i$, decreases counter $i$ and leads to a universal state.
This allows to accept exactly all words in $B$.
In addition, from any state $q$ of $M$, $\?A$ can move to a new countdown phase:
there is a transition
$q\step{a,\vec{0}}c$ to a new, final, control state
that can continue to read $a's$ while at least one of the counters remains non-zero.
This allows to accept exactly all words in $G$.
Note that the only non-determinism is for letters $\ztest_i$ when $M$'s $i$th counter after reading $\rho_i$ is not zero.
In this case, the only language-maximal choice is to move to the universal state.
The constructed system is therefore history-deterministic.

To conclude the proof, we argue that $L$ is regular iff $\rho$ visits only finitely many configurations.
Indeed, if so, then
$G$
is finite because all $x_i$, $i\le k$ are bounded,
and
$B$
is regular because at most $k$ many words $\incorr{k}$ exist.
So $L$ is the finite union of regular languages and thus regular.

Conversely, suppose that $M$'s run $\rho$ visits infinitely many different configurations. Then in particular,
there are infinitely many faithful prefixes $\rho_k$.
Let us assume towards contradiction that $L$ is regular and recognised by a DFA with $d$ many states. 
We pick a prefix $\rho_k$ so that $x_k>d$ and consider the word $\rho_ka^{x_k}\in L$.
While reading the suffix $a^{m}$, our DFA must repeat some cycle of length $c\le d$.
But then it must also accept $\rho_ka^{x_k+c}\notin L$ by going through that cycle twice.

The same proof goes through for the reachability semantics
if we set
$
G \eqdef \bigcup_{k\ge 0}
\left(
    \rho_k\cdot a^{=x_k}
\right)
$
and
$B \eqdef \bigcup_{k\ge 0}
\left(
    \incorr{k}\cdot\Sigma^{\ge x_k-1}
\right)
$.
Then again, if the run of $M$ visits finitely many configurations then both $G$ and $B$ are regular.
Otherwise $G$ is not regular.
The extra symbols at the end (of words in $G$ and $B$) allow a run of the VASS $\?A$ to decrease the counters to $0$ and accept
(and therefore to conclude that language $L=G\uplus B$ is in $\RHvass[2]$).
\end{proof}

\newpage

\appendix

\section{The Structure of Resolvers}
\label{app:lang-max}

We observe that any resolver $r$ must always make language-maximal choices. 
To formalise, let us write $\Post[a]{s,v}\eqdef \{(s',v')\mid (s,v)\step{a}(s',v')\}$ for the finite set of possible $a$-successor configurations of $(s,v)$.
Suppose a run produced by $r$ %
leads up to configuration $(s_i,v_i)$
and for the next letter $a_i$, it selects a continuation 
$(s_i,v_i)\step{a_i}(s_{i+1},v_{i+1})$.
Then $\Lang{s_{i+1},v_{i+1}} \supseteq
\bigcup\Lang{\Post[a_i]{s_i,v_i}}$.
When considering languages of finite words a useful observation is that
making only language maximal choices is not only necessary, but also a sufficient 
condition for $r$ to be a resolver.

\begin{restatable}{proposition}{propLangMax}
    \label{prop:hd=language-maximal}
    A function $r$ as above is a resolver iff
    all its choices are language maximal.
\end{restatable}

This does not depend on the finiteness of the state space.
A direct consequence is that resolvers can be assumed to be \emph{positional}:
That is, if any resolver $r$ exists then also one whose
decisions only depend on the current configuration and given letter,
not on the whole prefix run:
$r(\rho(s,v),a)=r(\rho'(s,v),a)$
for any two $\rho,\rho'\in (\states\x\N^k\x\Sigma)^*$ and letter $a$.
\begin{proof}
    Suppose a candidate resolver $r$ does not always make language-maximal choices.
    That is, for some word $w$ the corresponding run chosen by $r$ ends in some configuration
    $c$ and for some letter $a$, it moves to a successor configuration $c'$ that is not language maximal.
    Then there exist a suffix word $w'$ so that some run from $c$ on $aw'$ is accepting but no run from $c'$ on $w'$ is accepting, including the one chosen by $r$. So $r$ is not a resolver.

    Conversely, suppose a candidate resolver $r$ that always makes language maximal choices
    and assume towards a contradiction that it is not a resolver.
    This means that Player~1, wins the letter game from the initial configuration $c_0$: for some word $w=a_0a_1\ldots a_k\in \Lang{c_0}$ the run $c_0\step{a_0}c_1\step{a_1}\ldots \step{a_k}c_{k+1}$ constructed by $r$ is not accepting.
    Since some accepting run on $w$ exists, there must be a last configuration $c_j$ on this run
    which can still accept the suffix $w[j] = a_j a_{j+1}\ldots a_k$.
    This uses the assumption that we consider languages of finite words, not infinite ones.
    We conclude that the step $c_j\step{a_j}c_{j+1}$ was \emph{not} language maximal,
    since $w[j] \in \Lang{c_{j}}$ but $a_{j+1}\ldots a_k \notin \Lang{c_{j+1}}$.
    Contradiction.
\end{proof}

For \Ahvass[1] it is sufficient for a resolver to be semi-linear~\cite{PT2022}, meaning that, for each state and proposed letter, the counter configurations for which each available choice should be chosen can be expressed as a semi-linear set. However~\cite{PT2022} does not show the full power of semi-linear resolvers are required, and most natural examples appear to only require threshold queries (and often only to distinguish between zero and non-zero counter). We show that threshold comparisons with the counters is not sufficient: in the following example the system  must have  access to the parity of the counter.
\begin{figure}
\centering
\includegraphics{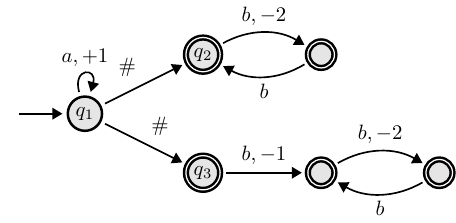}
\caption{A 1-\hdvass with a resolver that requires a resolver that depends on more than threshold comparisons.}
\label{fig:slresolver}
\end{figure}

\begin{example}
\label{eg:nonthreshold}

Consider the \Ahvass[1] depicted in \cref{fig:slresolver}. 
Observe that:
\begin{itemize}
\item $\Lang{q_1,0} = \{a^n\#b^m \mid m\le n\}$,
\item $\Lang{q_2,n} = \{ b^m \mid m\le n \text{ if }n \text{ even}  \text{ or } m \le n-1 \text{ if } n \text{ odd}  \}$,
\item $\Lang{q_3,n} = \{ b^m \mid m\le n \text{ if }n \text{ odd}  \text{ or } m \le n-1 \text{ if } n \text{ even}  \}$.
\end{itemize}
Hence, upon reading $\#$ a resolver must decide 
decide whether $\Lang{q_2,n} \subset \Lang{q_3,n}$ or $\Lang{q_3,n} \subset \Lang{q_2,n}$,
which is possible by looking at the parity of the counter value.
\end{example}

\section{Additional Material for Expressiveness (Section~\ref{sec:expressiveness})}
\label{app:expr}

\subsection{Comparison of History-deterministic VASS with and without \texorpdfstring{$\epsilon$}{epsilon}-transitions}

\label{app:epsrequired}

Given a number $n\in\N,n> 0$ let $bin(n)$ be the binary representation of $n$ in $1\{0,1\}^*$.
Consider the language 
$\LmustVASSe = \mathit{bin}(n)\#0^{\le n}\#$
, which is, for any number $n$, the binary representation of $n$ followed by $\#$, followed by at most $n$-many $0$'s. A result of~\cite{G1978} says that this language cannot be represented by any real-time machine, that is, machines without $\epsilon$-transitions. For completeness, we recall the argument in brief: observe that the value of the maximum counter can be at most $\norm{\delta}\log(n)$ after reading $bin(n)$, where $\norm{\delta}$ is the maximal counter effect of any transition. Therefore there are at most $poly(\log(n))$  configurations reachable after reading $bin(n)$, however, there are $2^{\log(n)-1}$ different numbers of length $|bin(n)|$. As a result, there are two numbers $n< m$ with $|bin(n)| = |bin(m)|$ (with $\log(n)$ large enough) for which reading either $bin(n)$ or $bin(m)$ has the same configuration; in which case either $bin(n)\#0^m$ is incorrectly accepted or $bin(m)\#0^n$ is incorrectly rejected.

\hdoneishdonee*
\begin{proof}
Clearly $\CHvass[1] \subseteq \CHvassE[1]$. We show that $\epsilon$ transitions can be removed from a $1\text{-}\hdvasse$. If there are no cycles then this is done in the standard way, merging them with the prior letter-consuming transitions. For cycles there are three cases. There are finitely many destinations on zero cycles and can be treated as in the acyclic case. Negative cycles are not beneficial, so the resolver should not iterate them. Therefore, we only add transitions necessary to access particular states, but keeping the counter maximal. Cycles with positive effect, for the purposes of maximal language acceptance, should be repeated infinitely. Thus it suffices to go to a copy of the automaton behaving only as a state-machine (without counter effects). Our procedure, adds finitely many counter maximal transitions.  We observe the new system is also HD, when reading a letter the resolver can decided where the resolver for the system would go with $a$ and then a sequence of $\epsilon$ transitions and move to a place with the same state and at least as high counter, which is language maximal.
\end{proof}

\section{Additional Material for Closure Properties (Section~\ref{sec:closure})}
\label{app:closure}

\lemunion*
\begin{proof}
 	 Let $L$ and $L'$ be recognised by by $\Ahvass[k]$ $\?A$ and $\Ahvass[k']$ $\?B$ respectively,
     in the coverability semantics. 
     W.l.o.g., assume that both are \emph{complete}, meaning there exists a (not necessarily accepting) run on every word.
     This can be guaranteed by adding new transitions to non-accepting sink states (which any resolver will avoid if possible).
 	The language $L\cup L'$ is accepted by the $\Ahvass[(k+k')]$ obtained by taking product $\?A\x\?B$, with $k+k'$ counters, where the first $k$ counters simulate the counters of $\?A$ and the last $k'$ counters simulate the counters of $\?B$. 
    A state in the product $(q,q')$ is accepting if either $q$ or $q'$ is accepting, or both are accepting.
 	The resolver for the product from a configuration $((q_1,q_1'), v)$ on a letter $a$ chooses the product of the transitions chosen by the resolver of $\?A$ and $\?B$ from the configurations $(q_1,v_k)$ and $(q'_1,v_{k'})$ respectively on the letter $a$, where $v_k$ and $v_k'$ are the projection of $v$ to the first $k$ and last $k'$ coordinates. 
 	
 	The construction works similarly for the intersection of $\CHvass$ and $\RHvass$ languages by taking accepting states $(q,q')$ in the product if both $q$ and $q'$ are accepting.
 	 \end{proof}

 We now consider closure of the language classes $\CHvass$ and $\RHvass$ for other operations, defined here shortly.
 The \emph{concatenation} of two languages $L$ and $L'$ is the languages $L\cdot L'= \{w=uv\mid u\in L,~ v\in L'\}$. Let $\Sigma$ and $\Gamma$ be  some alphabets and $h:\Sigma^*\rightarrow \Gamma^*$ be a homomorphism.
 The \emph{homomorphic image} of a language $L\subset\Sigma^*$ is $h(L)=\{h(w)\mid w\in L\} \subset \Gamma^*$. Similarly, the \emph{inverse homomorphic image} of a language $L\subset \Gamma^*$ is $h^{-1}(L) =
 \{w\in \Sigma^* \mid h(w) \in L\} $.
 
 Let $\Sigma=\{a_1,a_2,\dots,a_n\}$ be an alphabet. The \emph{Parikh image of a word} $w\in \Sigma^*$ is the vector $\Psi(w)=(v_1,v_2,\dots,v_n)$, where $v_i$ is the number of occurrences of $a_i$ in $w$. The \emph{Parikh image of a language} $L$ is the set of Parikh images of words in $L$.
For a language $L$, its \emph{commutative closure} $CC(L)$ is the language $\{w \mid \exists u\in L. \Psi(w)=\Psi(u)\}$.
\thmcovcl*
\begin{proof}
	The proof for closure under {union} and {intersection} is by \cref{lem:unionint}.
    For  inverse homomorphisms, let $L\subset \Gamma^*$ be in $\CHvass$ accepted by a $\Ahvass[k]$ $A$ with a resolver $r$. Let $Q$ be the set of states of $A$ and $\norm{\delta}$ be the largest absolute effect among all transitions.  
	Let $h:\Sigma^*\rightarrow\Gamma^*$ be a homomorphism and $\ell$ be such that $|h(a)|\leq \ell$, for all $a\in \Sigma$. Then $h^{-1}(L)$ is accepted by the $\Ahvass[k]$ $A'$ with the  states as $Q\times D^k$, where $D=[0,\ell\norm{\delta}]$. 	
	For every $a\in \Sigma$, a transition in $A'$ from $(q,v)$ to $(q',v')$ on $a$ will correspond to a run in $A'$ from $q$ to $q'$ on $h(a)$, so that the resolver $r'$ for $A'$ will simply choose the transition corresponding to the run chosen by resolver $r$ in $A$ on $h(a)$. 
	
	We need to show that $A'$ does not accept any word $w\not \in h_1(L)$. To show this, we need to ensure that if a run on $h(a)$ gets blocked due to some counter dropping below zero, the corresponding transition in $A'$ is also blocked. To do this, the transition in $A'$ has effect equal to the maximum negative effect in any prefix of the run on $h(a)$. The rest of the effect in the run on $h(a)$ is delayed to the next transition. Since the maximum effect is bounded by $\ell\norm{\delta}$, this can be stored in the states. The next transition will therefore have the sum of the effect delayed from the previous transition and the maximum negative effect in the prefix of the current transition. The details of the construction are below.
	
	Let $\rho=t_1t_2\dots t_{k'}$ be a path in $A$ from $q$ to $q'$ on $h(a)=\tfunction{\rho}$. Let $f_{ij}=\effect{t_1t_2\dots t_j}|_i$, i.e, the effect in the prefix up to $j$th transition projected to the $i$th counter. Let $f_i=\min_j (f_{ij})$ and $e_i=\min(f_i,0)$. Thus $e_i$ gives the largest negative effect in any prefix of the run. For every path $\rho$ on $h(a)$, we have a transition $((q,v),a,e',(q',v'))$ in $A'$ if $e'=e+v$, where $e=(e_1,e_2,\dots,e_k)$ and $v'+e=\effect{\rho}$.
	A state $(q,v)$ of $A'$ is initial if $q$ is initial in $A$ and $v=\vec{0}$ and $(q,v)$ in $A'$ is accepting if $q$ is accepting in $A$. 
	
    \medskip
    \noindent
	Now, we give counterexamples for the operations under which $\CHvass$ is not closed.		
    \begin{description}
        \item[Complementation.]
    Consider the language $\Lanblen=\{a^nb^{\leq n}\}$ which is in $\CDvass$. The complement of $\Lanblen$ is not even in $\CNvass$. Indeed if it were in $\CNvass$, then $\Lanblen^{c}\cap  a^{*}b^{*}= a^nb^{\ge n}$ would be in $\CNvass$, which is not the case.
        \item[Concatenation.]
    Consider the language $\LnotHDVASS= \Sigma^* \cdot a^nb^{\le n}$. By \cref{lemma:lastblocknotHD}, $\LnotHDVASS\not \in \CHvass$.
        \item[Homomorphisms.]
    Consider $L=\{c,d\}^* a^n b^{\le n}\in\CHvass$ which is accepted by even a $\Advass$. Let $h$ be the homomorphism $h(c)= a, h(d) = b$, which gives $h(L)=\LnotHDVASS\not \in \CHvass$ by \cref{lemma:lastblocknotHD}.
    \item[Kleene star.]
    Consider $(a^nb^{\leq n})^{*}$ which is the Kleene star of $\Lanblen\in \CHvass$. The proof of \cref{lemma:lastblocknotHD}, also shows that $\Lanblen^*$ is not in $\CHvass$.	
\item[Commutative closure.]
    Consider the commutative closure of $\Lanblen$, $L=CC(\Lanblen) = \{w \mid \#a \ge \#b\}$. If $L$ is in $\CHvass$, then $L\cap b^*a^* = b^na^{\ge n} $ is also in $\CHvass$ as $\CHvass$ is closed under intersection. However $b^na^{\ge n}$ is not even in $\CNvass$.
    \qedhere
\end{description}
\end{proof}

\thmreachcl*
\begin{proof}
	Closure under intersection follows from \cref{lem:unionint}.
	For the \emph{inverse homomorphic image}, a construction similar to the $\CHvass$, with states $(q,\vec{0})$ taken to be accepting in $A'$ for every $q$ that is accepting in $A$. 
	Note that any accepting run on $h(w)h(a)$, for any word $w\in \Sigma^*$ and $a\in \Sigma$, the effect of the run on $h(a)$ cannot be positive on any counter as it would lead to a non-zero counter value in the final configuration contradicting that the run is accepting. Therefore, the maximal negative effect encoded in the transition in our construction will always lead to a state $(q,\vec{0})$ and not delay any positive effect for later. This proves that $\RHvass$ are closed under inverse homomorphic image.
	
    \medskip
    \noindent
	Now, we give counterexamples for the operations under which $\RHvass$ is not closed.		
    \begin{description}
        \item[Unions.]
    Consider the language $\Lnotrunion=a^nb^{n} \cup a^n b^{2n}$. Both of the languages $a^nb^n$ and $a^nb^{2n}$ are in $\RHvass[1]$. Suppose $\Lnotrunion$ is recognised by a $\ARhvass[k]$ $A$. Since, $a^nb^n$ is in $L$, the resolver gives an accepting for all $n$, i.e, in a final state with all counters $0$. Let $n_1<n_2$ be such that the run given by resolver on $a^{n_1}b^{n_1}$ and $a^{n_2}b^{n_2}$ end in the same state $q$. 
	Since $a^{n_1}b^{n_1+n_1}$ is also accepted, the resolver extends the run from $(q,\vec{0})$ on the suffix $b^{n_1}$ and gives an accepting run. This also gives an accepting run on $a^{n_2}b^{n_1+n_2}$ which is a contradiction.
        \item[Complementation.]
    Consider $\Lanbgen = a^n b^{\ge n}$ which is in $\RHvass[1]$. Recall that $\Lanblen a^nb^{\le n} = \Lanbgen^c \cap a^*b^*$ is not in $\RHvass$ by \cref{lemma:rhvsrn}. If $\Lanbgen^c$ was in $\RHvass$, then so would $\Lanblen$ due to closure under intersection leading to a contradiction.
        \item[Concatenation.]
    Consider the concatenation of $a^*$ and $a^nb^{n}$, both in $\RDvass$, which gives the language $\Lanblen= a^nb^{\le n}$, which is not in $\RHvass$ by \cref{lemma:rhvsrn}. 
        \item[Homomorphisms.]
    Consider $\{c\}^* a^n b^{n}$ which is in $\RHvass[1]$ (even $\RDvass[1]$), and $h(c)= a$, gives the language $\Lanblen$ as above which is not in $\RHvass$ by \cref{lemma:rhvsrn}.
        \item[Kleene star.]
    Consider $\Lanbnastar = a^nb^na^*$, which is in $\RHvass[1]$. Indeed $\Lanbnastar^*$ is not in $\RHvass$. The run given by the resolver on $a^nb^na^m$ must end with $\vec{0}$, for all $m\ge0$. In particular, there exist $m_1<m_2$ such that the configuration reached on $a^nb^na^{m_i}$ are the same and therefore accept the same continuations. Thus $a^nb^na^{m_2} b^{m_2}$ cannot be distinguished from $a^nb^n a^{m_1} b^{m_2}$, which is a contradiction.
    \label{lang:lanbnastar}
        \item[Commutative closure.]
    $CC(a^n b^{\ge n}) \cap b^*a^* = b^{\ge n}a^n = b^n a^{\le n}$, which is not in $\RHvass$ by the same proof as \cref{lemma:rhvsrn} for $a^nb^{\le n}$ not being in $\RHvass$. 
            \qedhere
\end{description}
\end{proof}

To show that taking product for union ($\CHvass$) and intersection ($\CHvass$ and $\RHvass$) is not optimal in terms of number of counters, we have the following theorem. 

\begin{theorem}
	$\CHvass[1]$ is not closed under union and intersection. $\RHvass[1]$ is not closed under intersection.
\end{theorem}

\begin{proof}
	
	\begin{description}
		\item[Union.] 
		Consider the language $\Lnotcunion= a^nb^{\le n}c^* \cup a^nb^*c^{\le n}$, which is the union of two languages in $\CHvass[1]$. Suppose $\Lnotcunion$ is in \CHvass[1].
	Let $|Q|$ be the number of states and $\norm(\delta)$ be the maximum counter effect of transitions in the $\Ahvass[1]$ accepting $a^nb^{\le n}c^*\cup a^nb^*c^{\le n}$. 
	Consider the sequence of words $w_n=a^nb^n`bc^n$ and the runs given on these words by the resolver. Let $(q_n,v_n)$ be the configuration reached after reading $a^nb^n$, for every $n$. Now, we consider two cases. Suppose there exists a bound $B$ such that $v_n<B$ for all $n$. Then there exists $n_1<n_2$ such that $(q_{n_1},v_{n_1})=(q_{n_2},v_{n_2})$. Since $bc^{n_2}\in L(q_{n_2},v_{n_2})$, we get an accepting run on $a^{n_1}b^{n_1}bc^{n_2}$ which is a contradiction.
	Therefore, the counter values $v_n$ must be unbounded.
	
	Now, consider the infinite sequence of words $w_{n_1}, w_{n_2}, \dots$ such that the state reached after $a^{n_i}b^{n_i}$ is the same, i.e, $q_{n_1}=q_{n_2}=\dots$ and the last $|Q|$ many transitions leading to $(q_{n_i},v_{n_i})$ are also the same. Note that since there are finitely many choices of the last state and $|Q|$ length sequence of transitions, such an infinite subsequence must exist.
	Let $(q^{j}_{n_i},v^{j}_{n_i})$ denote the configuration reached in the run on $a^{n_i}b^{n_i}$ given by resolver after the prefix $a^{n_i}b^{n_i-j}$, for $j\le |Q|$. 
	It is easy to see that $v_{n_1}<v_{n_2}<\dots$, as $L(q_{n_i},v_{n_i})\subsetneq L(q_{n_{i'}},v_{n_{i'}})$ for $i<i'$ witnessed by $bc^{n_{i'}}\in L(q_{n_{i'}},v_{n_{i'}})$ but not in $L(q_{n_{i}},v_{n_{i}})$.
	Since the last $|Q|$ transitions leading to $(q_{n_i},v_{n_i})$ are the same, we can also conclude that $v^j_{n_1}<v^j_{n_2}<\dots$, for all $j\le |Q|$. We write $q^j$ to denote $q^j_{n_i}$ since the state is the same for all choices of $i$.
	
	Note that $b^jc^*\subseteq L(q^j,v^j_{n_i})\not \supseteq b^{j+1}c^*$, for all $n_i$, $j\le |Q|$. The inclusion of $b^jc^*$ is immediate because the resolver must make language maximal choices. However, if $b^{j+1}c^*$ is included, then we get an accepting run on $a^{n_i}b^{n_i+1}c^{n_i+i}$, which is a contradiction.
	This shows that $b^jc^*\subseteq L(q^j,c)\not \supseteq b^{j+1}c^*$ for any $c>v^j_{n_1}$. This is because languages from configurations with the same state are monotone in the value of the counter.

	Note that there exists a $j<j'$ such that $q^j=q^{j'}$ as we look at runs whose last $|Q|$ transitions (and therefore $|Q|+1$ states) are the same. Now choose $n_i$ such that $v^{j}_{n_i}$ and $v^{j'}_{n_i}$ are both bigger than $min(v^{j}_{n_1},v^{j}_{n_1})$. Therefore, by the previous observation,  
	$b^jc^*\subseteq L(q^{j'},v^{j'}_{n_i})\not \supseteq b^{j+1}c^*$.  This means $b^{j'}c^*$ is not accepted from $(q^{j'},v^{j'}_{n_i})$ which contradicts that the run was chosen by a resolver. This concludes the proof that the language $\Lnotcunion$ is not accepted by any $\Ahvass[1]$.

	\item[Intersection.]
	 Consider the language $\Lnotcint=a^nb^{\le n}c^* \cap a^nb^*c^{\le n}$ and suppose it is accepted by $\Ahvass[1]$. Consider the runs on $a^nb^nc^n$ given by the resolver. If the configuration reached after $a^nb^n$ has counter value bounded, then by a similar reasoning to the union case, we can find $n_1<n_2$ such that the configuration reached by the resolver after reading $a^{n_1}b^{n_1}$ and $a^{n_2}b^{n_2}$ are the same and we get an accepting run on $a^{n_1}b^{n_1}c^{n_2}$ which is a contradiction. If the configuration is unbounded, we get a $n$ such that the configuration reached after reading $a^nb^n$ has counter value $>(|Q|+1)\norm(\delta)$. This allows to repeat a cycle in $b^n$ block as the maximum decreasing effect is at most $(|Q|+1)\norm(\delta)$. This gives an accepting run on $a^nb^m$, where  $m>n$, which is a contradiction.
	
	For the reachability semantics, the proof is even simpler as $\Lnotrint=a^nb^nc^n$ is not even context-free and therefore not definable even with zero tests on the counter.
	\qedhere
\end{description}
	
\end{proof}

\section{Additional Material for Decision Problems (Section~\ref{sec:decision})}
\label{app:decision}

\cmfinite*
\begin{proof}
	Suppose that one could decide above question. Then one can also decide the halting problem:
	if the set of reachable configurations is infinite then clearly $M$ does not halt.
	Otherwise, we can determine if $M$ halts by simulating it either until it halts, or if it re-visits one configuration without halting.
\end{proof}

\section{Index of Languages used in this paper}
\label{sec:index}
\begin{center}
\begin{tabularx}{\textwidth}{ |l|X|c|c| }     
 \hline
 
 Name    & Definition & Alphabet &Page \\
  \hline
 $\LnotDVASS$  
    & $ a^nb^{\le n} + a^*b^*c$
    & $\{a,b,c\}$
    & \pageref{lang:LnotDVASS}\\
 $\Lanbgen$
    & $a^nb^{\ge n}\#$
    & $\{a,b,\#\}$
    & \pageref{lang:Lanbgen}\\
 $\LnotHDVASS$  
    & $(a+b)^*a^nb^{\le n}$
    & $\{a,b\}$
    & \pageref{lang:LnotHDVASS}\\
 $\Lanblen$
    & $a^nb^{\le n}$
    & $\{a,b\}$
    & \pageref{lang:Lanblen}\\
 $\Lanbn$
	& $a^nb^n$
 	& $\{a,b\}$
 	& \pageref{lang:Lanbn}\\
 $\LmustVASSe$
 	& $\mathit{bin}(n)\#0^{\le n}\# $, where $\mathit{bin}(n)$ is $n$ in binary.
 	& $\{0,1,\#\}$
 	& \pageref{lang:LmustVASSe}\\
 $\Lanblenbarrier$
 	& $a^nb^{\le n}\#$
 	& $\{a,b,\#\}$
 	& \pageref{lang:Lanblenbarrier}\\ 	
 $\LnotFSVASS$  
    & $\bigcup_{k=0}^{\infty} a^{n_0}b^{n_0}\dots a^{n_{k-1}}b^{n_{k-1}}a^{n_{k}} b^{\le n_{k}}a\Sigma^*$.
    & $\{a,b\}$
    & \pageref{lang:LnotFSVASS}\\

    $\Lnotrunion$
    & $a^nb^n \cup a^nb^{2n}$
    & $\{a,b\}$
    & \pageref{lang:lnotrunion}\\ 
   	$\Lnotcunion$
    & $a^nb^{\le n}c^*\cup a^nb^*c^{\le n}$
    & $\{a,b,c\}$
    & \pageref{lang:lnotcunion}\\
        $\Lnotcint$
    & $a^nb^{\le n}c^n \cap a^nb^*c^{\le n}$
    & $\{a,b,c\}$
    & \pageref{lang:lnotcint}\\
        $\Lnotrint$
    & $a^nb^nc^n$
    & $\{a,b,c\}$
    & \pageref{lang:lnotrint}\\    
    $\Lanbnastar$
    & $a^nb^na^*= a^nb^nc^* \cap a^nb^*c^n$
    & $\{a,b\}$
    & \pageref{lang:lanbnastar}\\

 \hline
\end{tabularx}
\end{center}
 \end{document}